\documentclass [journal]{IEEEtran}

\usepackage{amsfonts,amsmath,amssymb}
\usepackage{indentfirst}
\usepackage{url,float}
\usepackage{color}
\usepackage[a4paper,ignoreall]{geometry}
\usepackage{longtable}
\usepackage{graphicx}
\usepackage[a4paper,ignoreall]{geometry}
\usepackage{multicol}
\usepackage{stfloats}

\def\qu#1 {\fbox {\footnote {\ }}\ \footnotetext { From Qu: {\color{red}#1}}}
\def\hqu#1 {}
\def\yin#1 {\fbox {\footnote {\ }}\ \footnotetext { From Yin: {\color{blue}#1}}}
\def\hyin#1 {}

\newtheorem{result}{Result}
\newtheorem{Th}{Theorem}[section]
\newtheorem{Cor}[Th]{Corollary}
\newtheorem{Prop}[Th]{Proposition}
\newtheorem{Prob}[Th]{Problem}
\newtheorem{Lemma}[Th]{Lemma}
\newtheorem{Def}[Th]{Definition}
\newtheorem{example}{Example}

\newcommand{\tr}{{\rm Tr}}

\newcommand{\gf}{{\mathbb F}}

\makeatletter
\newcommand{\figcaption}{\def\@captype{figure}\caption}
\newcommand{\tabcaption}{\def\@captype{table}\caption}
\makeatother

\begin{document}


\title{A New Approach to Constructing  Quadratic Pseudo-Planar Functions over $\gf_{2^n}$ }



\author{Longjiang~Qu
        \thanks{L.J. Qu is with
                College of Science, National University of Defense Technology, ChangSha, 410073,
                China. He is also with State Key Laboratory of Cryptology, P. O. Box 5159, Beijing, 100878, China.
                Email: ljqu\_happy@hotmail.com.
                The research of L.J. Qu was supported by the NSFC of China
                under Grant  61272484, 11531002, 61572026,  the National Basic
Research Program of China(Grant No. 2013CB338002), the Basic Research Fund of National University of Defense Technology (No. CJ 13-02-01),
the Open Foundation of State Key Laboratory of Cryptology
and the Program for New Century Excellent Talents in University (NCET).}
       }

\maketitle{}

\begin{abstract}
Planar functions over finite fields give rise to finite projective planes. They
were also used in the constructions of DES-like iterated ciphers, error-correcting codes, and codebooks.
They were originally defined only in finite fields with odd characteristic, but recently Zhou introduced pesudo-planar functions
in even characteristic which yields similar applications.
All known pesudo-planar functions are quadratic  and hence they give presemifields. In this paper, a new
approach to constructing quadratic pseudo-planar functions is given. Then five explicit families of
pseudo-planar functions are constructed, one of which is a binomial,
two of which are trinomials, and the other two are quadrinomials. All known pesudo-planar functions
 are revisited, some of which are generalized.
These functions not only lead to projective planes,
relative difference sets and presemifields,  but also give optimal codebooks meeting the Levenstein bound,
complete sets of mutually unbiased bases (MUB) and compressed sensing matrices with low coherence.
\end{abstract}


\section{Introduction}

\IEEEPARstart{L}et $p$ be an odd prime and $n$  a positive integer. A function
$F:\gf_{p^n}\rightarrow \gf_{p^n}$ is \emph{planar} if the mapping
$$x\mapsto F(x+a) - F(x)$$
is a permutation of $\gf_{p^n}$ for each $a\in \gf_{p^n}^\ast$, where
$\gf_{p^n}^\ast$ denotes the set of all nonzero elements of $\gf_{p^n}$. Planar functions
were introduced by Dembowski and Ostrom to construct finite projective planes and
arised in many other contexts. For example, Ganley and Spence \cite{GanleySpence} showed that
planar functions give rise to certain relative difference sets, Nyberg and
Knudsen \cite{NybergKnudsen}, among others, studied planar functions
for applications in cryptography, Carlet,
Ding, and Yuan \cite{CDY}, among others, used planar functions to construct error-correcting
codes, and Ding, and Yin \cite{DingYin}, among others, used planar functions to construct optimal codebooks
meeting the Levenstein bound.

If $p=2$, then there are no planar functions $F:\gf_{p^n}\rightarrow \gf_{p^n}$ since $0$ and $a$ have
the same image under the map $x\mapsto F(x+a) - F(x)$. Recently, Zhou \cite{Zhou2013}
 introduced a characteristic $2$ analogue of planar functions, which have the same types of applications as do
 odd-characteristic planar functions.

\begin{Def}\label{def_PN}
A function $F:\gf_{2^n}\rightarrow \gf_{2^n}$ is called \emph{pseudo-planar} if
\begin{equation}\label{eq_def_PN}
    F(x+a)+F(x)+ax
\end{equation}
is a permutation polynomial over $\gf_{2^n}$ for each $a\in \gf_{2^n}^\ast$.
\end{Def}

Note that Zhou \cite{Zhou2013} called such functions ``\emph{planar}'',
and the term ``\emph{pseudo-planar}'' was first used by Abdukhalikov \cite{Abdukhalikov} to avoid confusion with
planar functions in odd characteristic. Schmidt and Zhou \cite{SchmidtZhou} showed a pseudo-planar function
can be used  to produce a finite projective plane, a relative difference set with parameters $(2^n, 2^n, 2^n, 1)$,
and certain codes with unusual properties. Abdukhalikov \cite{Abdukhalikov} used pseudo-planar functions
to give new explicit constructions of complete sets of MUBs, and showed the connection between quadratic
pseudo-planar functions and commutative presemifields. Here, as usual,
 a quadratic function refers to a function with algebraic degree $2$,
 which is also called a Dembowski-Ostrom type function.
  It should be noted that we distinguish \emph{algebraic degree} and \emph{degree} in this paper. Let
$F(x)=\sum_{i=0}^{2^n-1} c_ix^i$ be a polynomial over $\gf_{2^n}$. Then its \emph{algebraic degree} is defined to
be the maximum $2$-adic weight of $i$ for all nonzero $c_i$, while its \emph{degree} is defined to
be the maximum integer $i$ for all nonzero $c_i$. For example, the algebraic degree of $x^6$ is $2$,
while its degree  is $6$. A function with algebraic degree at most $1$ is called a \emph{linearized polynomial}.
It is trivial that a linearized polynomial is necessarily pseudo-planar. It is also clear that
a function  is pseudo-planar if and only if so is the summation of it with any linearized polynomial.
Hence, throughout this paper, we assume that a function is free of linearized terms, that is, the coefficient of
$x^{2^i}$ is $0$ for any nonnegative integer $i$.

To the best of  the author's knowledge, all known pesudo-planar functions are of Dembowski-Ostrom type.
The equivalence on them is the same as the isotopism of the corresponding semifields.
(See Section II.A for more details.) Moreover, there are only two types of presemifields with
even characteristic, that is, finite fields and the Kantor family
\cite{LavrauwPolverino}\cite{Zhou2013}.

\begin{result} \label{result1} \cite[Examples 2.1 and 2.2]{Zhou2013}

1) For each positive integer $n$, every affine mapping, especially $f(x) =
0$, is a pesudo-planar function on $\gf_{2^n}$. The corresponding plane is a Desarguesian plane
and the corresponding semifield is the finite field.

2) Assume that we have a chain of fields $\gf = \gf_0 \supset \gf_1 \supset \cdots \supset \gf_r$ of
characteristic $2$ with $[\gf : \gf_r]$ odd and corresponding trace mappings $\tr_i : \gf \rightarrow\gf_i$.
Then
  \begin{equation}\label{eq_Kantor_Pla}
    \left(x\sum_{i=1}^r \tr_i(\zeta_i x)\right)^2, \text{where}\ \zeta_i\in  \gf^\ast
\end{equation}
 is a pesudo-planar function on $\gf$, which is corresponding to the Kantor family of commutative presemifields \cite{Kantor}.
\end{result}


It seems to be quite difficult to find pesudo-planar functions which are inequivalent to those in Result \ref{result1}.
Schmidt and Zhou \cite{SchmidtZhou},
 and Scherr and Zieve \cite{ScherrZieve}  turned to study the classification of monomial planar functions.
 Three families of  monomial pseudo-planar functions were got.
 However, as  pointed out by Schmidt and Zhou, the corresponding planes are all desarguesian, i.e., the semifields are finite fields, or the functions are all equivalent to $F(x)=0$.

 \begin{result}\label{result_mon}
 The following monomials are pesudo-planar functions.

 1) $F(x)=cx^{2^m}$, where $c\in \gf_{2^n}$ (Trivial);

 2)  $F(x)=cx^{2^m+1}$, where $n=2m$, $c\in \gf_q^\ast$ and $\tr_{m/1}(c)=0$ and $\tr_{m/1}$ denotes the trace function from $\gf_{2^m}$ to $\gf_2$
 (\cite[Theorem 6]{SchmidtZhou},
 generalized by Theorem \ref{th_ind2Monomial});

 3) $F(x)=cx^{2^{2m}+2^{m}}$, where  $n=3m$, $m$ is even, $q=2^m$, $c\in \gf_{2^n}^\ast$ is
                              a ($q-1$)-th power but not a $3(q-1)$-th power
                                (\cite[Theorem 1.1]{ScherrZieve}, see also Proposition \ref{prop_ind3Monmial}).
 \end{result}

Later,  Hu, Li, Zhang, et. al. \cite{HuLiZhang} introduced three families of binomial pesuso-planar functions.

 \begin{result}\label{result_binom}
  The following binomials are pesudo-planar functions.

  1)  $F(x)=a^{-(q+1)}x^{q+1} + a^{q^2+1}x^{q^2+1}$, where $n=3m$, $q=2^m$ and $a$ satisfies a trace equation
 (see \eqref{eq_new_condition} or \eqref{eq_HLZ_condition} in Example \ref{ex_known_Pla}.(3) ) (\cite[Proposition 3.2]{HuLiZhang}).

 2) $F(x)=x^{q+1} + x^{q^2+q}$, where $n=3m$, $m\not\equiv 2\mod 3$, and  $q=2^m$ (\cite[Proposition 3.6]{HuLiZhang}).

 3) $F(x)=x^{q^2+q} + x^{q^2+1}$, where  $n=3m$, $m\not\equiv 1\mod 3$, and  $q=2^m$ (\cite[Proposition 3.8]{HuLiZhang}).
 \end{result}

It is open to classify the pseudo-planar
functions. Only the classification of the monomial pseudo-planar functions was studied, and it was conjectured
that there are only three families of such monomials \cite[Conjecture 3.2]{SchmidtZhou}.

Throughout the rest of this section, let $n=tm$, and let $q=2^m$, where $t, m$ are positive integers and $t\geq 2$.
Then $\gf_{2^n}$  is an extension  field of $\gf_{2^m}$  with extension degree $t$.

There are five families of pseudo-planar functions excluding the trivial monomial one in Results \ref{result_mon} and \ref{result_binom}.
Four families of them are defined over $\gf_{2^{3m}}$,
 and the rest one is defined over $\gf_{2^{2m}}$. Further,
 all the exponents of the terms in these five families are in the set of $\{q^2+q, q^2+1,  q+1\}$,
 where $q=2^m$.

In this paper, a new approach to constructing quadratic pseudo-planar functions is introduced.
Firstly, according to Definition \ref{def_PN}, a quadratic function $F$ over $\gf_{2^n}$ is pseudo-planar
if and only if $$\mathbb{L}_a(x) := F(x+a)+F(x)+F(a)+ax$$ is a linearized permutation polynomial for each $a\in \gf_{2^n}^\ast$.
We then convert it to studying the permutation property of
 the dual polynomial $\mathbb{L}_{b}^\ast(a)$ (see the proof of Theorem \ref{th_Gen} for the detailed definition) of $\mathbb{L}_a(x)$,
 and further link it with the problem of deciding whether a corresponding determinant can be zero.
 For the general family of functions defined by \eqref{eq_gen_Form} (in Theorem \ref{th_Gen}), this
 determinant is of size $t$, and with additional properties which will simplify the later calculation.
 Secondly, we relate this determinant with a polynomial $m_b(x)$ (cf. \eqref{eq_mb_Gen} in Section III.B) over $\gf_q$ with degree $t$.
Assuming the determinant to be zero leads to an equation on the coefficients of $m_b(x)$.
Then the problem is reduced to discussing whether there exists an irreducible polynomial $m_b(x)$ over $\gf_q$ satisfying
the aforementioned equation. Please refer to Section III for more details.

Then we use this new approach to construct new
explicit families of quadratic pseudo-planar functions over $\gf_{2^n}$,
and reconstruct known families.
The constructions are split into three cases according to the values of the extension degree $t$.
For the  case of extension degree $t=3$,
we construct three families of pseudo-planar functions, and study a family of trinomial, which
is a generalization of the three families of functions in \cite{HuLiZhang}.
The monomial polynomial is also revisited, and a sufficient and necessary condition for it to be pseudo-planar is given.
For the case of extension degree $t=4$, we construct two families of pseudo-planar functions.
One is a trinomial, the other is a quadrinomial.
For the case of extension degree $t=2$, we revisit the monomial pseudo-planar function and provide a simple
sufficient and necessary condition, which generalize \cite[Theorem 6]{SchmidtZhou}. However,
we can not construct pseudo-planar function with new explicit form in this case and leave it as an open problem.
The equivalence problem of these constructed functions is then investigated.
%
%
The  functions constructed in this paper not only lead to projective planes,
relative difference sets and presemifields,  but also give optimal codebooks meeting the Levenstein bound,
complete sets of MUBs and compressed sensing matrices with low coherence.

The rest of this paper is organized as follows. Necessary
definitions and results are given in Section II. In Section III we
introduce the new approach of constructing quadratic pesudo-planar functions. Several families of such functions
with new forms are constructed in Section IV, which is divided into three subsections according to
the values of the extension degree $t$. In Section V,  the equivalence problem of these functions
is investigated. A small  application example is given in Section VI. Section VII is  the concluding remarks.

\section{Preliminaries}

In this section, we  give necessary definitions and results which
will be used in the paper.

\subsection{Relative Difference Set, Galois Ring and Presemifield}

Let $G$ be a finite abelian group and let $N$ be a subgroup of $G$. A subset $D$ of $G$ is a \emph{relative
difference set (RDS)} with parameters $(|G|/|N|, |N|, |D|, \lambda)$ and \emph{forbidden subgroup} $N$ if the
list of nonzero differences of $D$ comprises every element in $G\setminus N$ exactly $\lambda$ times, and no
element of $N\setminus \{0\}$. We are interested in RDSs $D$ with parameters $(q, q, q, 1)$ and a normal
forbidden subgroup, in that case a classical result due to Ganley and Spence \cite[Theorem 3.1]{GanleySpence}
 shows that $D$ can be uniquely extended to a finite projective plane.
Particularly, if   $D$ is with parameter
$(2^n, 2^n, 2^n, 1)$, then $D$ is necessarily a subset of $\mathbb{Z}_{4}^n$ (where $\mathbb{Z}_4=  \mathbb{Z} /4\mathbb{Z}$) and
the forbidden subgroup is $2\mathbb{Z}_4^n$. This fact motivated Zhou \cite{Zhou2013} to study such difference
sets, which then led to the notion of pseudo-planar functions over finite fields of characteristic two.

We recall some basic facts about the Galois ring $R=GR(4^n)$ of characteristic $4$ and cardinality
$4^n$. We have
$R/2R\cong \gf_{2^n}$, the unit group $R^\ast = R\setminus 2R$ contains a cyclic subgroup $C$
of size $2^n-1$ isomorphic to $\gf_{2^n}^\ast$. The set $\mathcal{T}=\{0\}\cup C$ is
called the \emph{Teichm$\ddot{\text{u}}$ller set} in R. Every element $x\in R$ can be
written uniquely in the form $x=a+2b$ for $a, b\in \mathcal{T}$.
Then \emph{the trace function over Galois ring $R$} is defined as follows.
$$\tr_R (x) = (a + a^2 + \cdots + a^{2^{n-1}}) + 2 (b + b^2 + \cdots + b^{2^{n-1}}).$$
Since $R/2R\cong \gf_{2^n}$, for every element $u\in \gf_{2^n}$ there exists a corresponding
unique element $\widehat{u} \in \mathcal{T}$, called the \emph{Teichm$\ddot{\text{u}}$ller lift} of $u$.
Using the  Teichm$\ddot{\text{u}}$ller lift, we can also regard a function $F:\gf_{2^n}\rightarrow \gf_{2^n}$
as a function $F:\mathcal{T}\rightarrow \mathcal{T}$. For more information on Galois rings, please
refer to \cite{HKCSS94}\cite{Wan2003}.

It can be easily proved that a relative difference set in $R$ with parameters $(2^n, 2^n, 2^n, 1)$
can always be written as
\begin{equation}\label{eq_RDS}
    D=\{x+2\sqrt{F(x)} : x \in \mathcal{T}\},
\end{equation}
where $F$ is some function from $\mathcal{T}$ to itself, and $\sqrt{x}$ denotes $x^{2^{n-1}}$.
Then we have the following link between RDS in $R$ and pseudo-planar functions over $\gf_{2^n}$.

\begin{Th}\cite[Theorem 2.1]{SchmidtZhou}
The set $D$, given in \eqref{eq_RDS} is a relative difference set in $R$ with parameters $(2^n, 2^n, 2^n, 1)$
and forbidden group $2R$ if and only if $F$ is pseudo-planar over $\gf_{2^n}$.
\end{Th}

A \emph{presemifield} is a ring with no zero-divisor, and with left and right distributivity \cite{Dembowski}. A presemifield
with multiplicative identity is called a \emph{semifield}. A finite presemifield can be
obtained from a finite field $(\gf_{q}, +, \cdot)$ by introducing a new product operation $\star$, so
it is denoted by $(\gf_{q}, +, \star)$.
An \emph{isotopism} between two presemifields $P = (\gf_q, +, \star)$ and $P' = (\gf_q, +,  \circ)$ is a
triple $(M, N, L)$ of  bijective linearized mapping $\gf_q \rightarrow \gf_q$ such that
$$ M(x) \circ N(y) = L(x \star y), \text{ for all } x, y \in \gf_q.$$
Any presemifield $P = (\gf_q, +, \star)$ is isotopic to a semifield: fix any $0\neq e\in \gf_q$ and
define $\circ$ by $(x \star e)\circ  (e \star y) = x \star y$ for all $x, y \in \gf_q$. Then $(\gf_q, +, \circ)$ is a semifield with identity $e\star e$, and is obviously isotopic to $P$. If $(\gf_q, +, \star)$ is commutative then
so is each such semifield $(\gf_q, +, \circ)$.

There exists a correspondence between commutative
semifield (up to isotopism) over finite fields of characteristic two and quadratic  pseudo-planar functions \cite[Theorem 9]{Abdukhalikov}.
More specifically, if $F$ is a quadratic pseudo-planar function over $\gf_{2^n}$, then
$(\gf_{2^n}, +, \star)$ with multiplication $x \star y =
xy + F (x + y) + F (x) + F (y)$ is a presemifield. On the other side, if $(\gf_{2^n}, +,  *)$ is a commutative presemifield, then there exist a strongly isotopic
commutative presemifield $(\gf_{2^n}, +, \star)$ and a pseudo-planar function $F$ such that
$x \star y =
xy + F (x + y) + F (x) + F (y)$.

Let $\mathbb{S}=(\gf_{p^n}, +, \ast)$ be a semifield. The subsets
  $$N_l(\mathbb{S}) = \{a\in \mathbb{S} | (a\ast x)\ast y = a\ast (x\ast y) \ \text{ for all} \ x,y \in \mathbb{S} \},$$
  $$N_m(\mathbb{S}) = \{a\in \mathbb{S} | (x\ast a)\ast y = x\ast (a\ast y) \ \text{ for all} \ x,y \in \mathbb{S} \},$$
  $$N_r(\mathbb{S}) = \{a\in \mathbb{S} | (x\ast y)\ast a = x\ast (y\ast a) \ \text{ for all} \ x,y \in \mathbb{S} \},$$
are called the \emph{left, middle and right nucleus} of $\mathbb{S}$, respectively. It is easy to
check that these sets are finite fields.

A pseudo-planar function is just a field-function illustration of the $(2^n,2^n,2^n,1)$-RDS in $\mathbb{Z}_4^n$,
 and the equivalence of RDSs in $\mathbb{Z}_4^n$ is the same as the isotopism of the corresponding semifields \cite[Proposition 3.4]{Zhou2013}.
 Hence if the pseudo-planar functions are of Dembowski-Ostrom type, then the equivalence on them is the same as the isotopism of the corresponding semifields.
 To check  whether a  semifield is  new or not, a natural way is to determine its left (right) nucleus.


\subsection{Codebook,  MUB and Compressed Sensing Matrix}

Let $\mathcal{C} = \{\mathbf{c}_0, \cdots, \mathbf{c}_{N-1}\}$, where each $\mathbf{c}_l$ is a unit norm
$1\times K$ complex vector over an alphabet $A$. Such a set $\mathcal{C}$ is
called an $(N,K)$ \emph{codebook} (also called a signal set). The size
of $A$ is called the \emph{alphabet size} of $\mathcal{C}$. As a performance
measure of a codebook in practical applications, the \emph{maximum
crosscorrelation amplitude} of an $(N,K)$ codebook  $\mathcal{C}$  is defined
by
$$
I_{\max}( \mathcal{C} ) = \max_{0\leq i<j\leq N-1} |\mathbf{c}_i\mathbf{c}_j^H|,
$$
where $\mathbf{c}^H$ stands for the conjugate transpose of the complex
vector $\mathbf{c}$. For $I_{\max}( \mathcal{C} )$, we have the well-known
\emph{Welch bound} \cite{Welch1974} and the \emph{Levenstein bounds} \cite{KabatyanskiiLevenshtein}\cite{Levenshtein}, while the latter are better than
the former when $N$ is large.
For latter use, we give in the following the Levenstein bound for complex-valued codebooks.

\begin{Lemma}(Levenstein Bound)
For any complex-valued $(N,K)$ codebook $\mathcal{C}$ with $N> K^2$,
we have
\begin{equation}\label{eq_Lev_boundC}
    I_{\max}( \mathcal{C} ) \geq \sqrt{\frac{2N-K^2-K}{(K+1)(N-K)}}.
\end{equation}
\end{Lemma}

Constructing codebooks achieving the Welch bound or the Levenstein bound
looks very hard in general.
An efficient approach is to use combinatorial objects such
as difference sets, almost difference sets, and
so on (see \cite{Ding06}\cite{DingFeng07}\cite{DingFeng08}\cite{ZhangFeng} and the references therein).
Particularly, Zhou and Tang used relative difference sets to
construct codebooks \cite{ZT2011}.

Let $G$ be a finite abelian group and let $N$ be a subgroup of $G$ with order $v$
and index $u$. Set $\hat{G}$ be the set of all the characters of $G$. Let
$D=\{d_0, \cdots, d_{k-1}\}$ be a $k$-subset of $G$. For any $\chi \in \hat{G}$,
we define a complex codeword
$$\mathbf{C}_{\chi}=\frac{1}{\sqrt{k}}(\chi(d_0), \cdots, \chi(d_{k-1}) ).$$
Then we define the codebook
\begin{equation}\label{eq_Codebook}
    \mathcal{C}_{D}=\{\mathbf{C}_{\chi}: \chi \in \hat{G}\}\cup E_{k},
\end{equation}
where $E_{k}=\{e_{i}: 1\leq i\leq k\}$ is the standard basis of the $k$-dimensional Hilbert space.

\begin{Th}\cite[Theorem 3.1]{ZT2011}
Let $D$ be a $(u, v, k, \lambda)$ relative difference set in $G$ relative to $N$.
Then $\mathcal{C}_{D}$ of \eqref{eq_Codebook} is a $(uv+k, k)$ codebook with $I_{\max}(\mathcal{C}_{D})=\sqrt{\frac{1}{k}}$.
\end{Th}

In particular, we have the following corollary.
\begin{Cor}\label{Cor_codebook}
Let $D$ be a $(q, q, q, 1)$ relative difference set in $G$ relative to $N$.
Then $\mathcal{C}_{D}$ of \eqref{eq_RDS} is a
$(q^2+q, q)$ codebook with $I_{\max}(\mathcal{C}_{D})=\sqrt{\frac{1}{q}}$,
 which is an optimal codebook meeting the Levenstein bound \eqref{eq_Lev_boundC}.
\end{Cor}

For $q$ odd, a $(q, q, q, 1)$ RDS is corresponding to a planar function over $\gf_q$.
Optimal codebooks from planar functions were originally presented by Ding and Yin \cite{DingYin}.
However, for $q$ even, pseudo-planar functions and the corresponding optimal codebooks
seem not to be widely known by the codebook researchers.
For  others (known) codebooks meeting the
Levenshtein bound, please refer to  \cite{XDM}\cite{ZDL14} and the references therein.

To write explicitly the codebook from a pseudo-planar function, one need to write explicitly
the characters over the underlying group, the additional group of the Galois ring $GR(4^n)$.
This was done by K. Abdukhalikov in the language of \emph{mutually unbiased base (MUB)} \cite{Abdukhalikov}.
A set of MUBs in the Hilbert space $\mathbb{C}_n$ is defined as
a set of orthonormal bases $\{B_0, B_1, \cdots, B_r\}$ of the space such that the square of the
absolute value of the inner product $|(x, y)|^2$ is equal to $1/n$ for any two vectors $x, y$
 from distinct bases. Mutually
unbiased bases have important applications in quantum physics \cite{WoottersFields}.
Recently it was discovered that MUBs are very closely related or
even equivalent to other problems in various parts of mathematics, such as algebraic
combinatorics, finite geometry, discrete mathematics, coding theory, metric geometry,
sequences, and spherical codes.

There is no general classification of MUBs. The main open problem in this area is
to construct a maximal number of MUBs for any given $n$. It is known that the maximal
set of MUBs of $\mathbb{C}_n$ consists of at most $n + 1$ bases, and sets attaining this bound are
called complete sets of MUBs. Constructions of complete sets of MUBs  are known only for
prime power dimensions. Even for the smallest non-prime power dimension six the
problem of finding a maximal set of MUBs is extremely hard  and remains open
after more than 30 years. For known constructions of MUBs and their link with
the complex Lie algebra $sl_n(\mathbb{C})$, please refer to \cite{Abdukhalikov}
and the references therein. Particularly, it was shown that
pseudo-planar functions over $\gf_{2^n}$ can be used to construct complete sets of
MUBs in $\mathbb{C}^{2^n}$.

\begin{Th}\cite[Theorem 8]{Abdukhalikov}\label{th_Abdukhalikov}
Let $F$ be a pseudo-planar function over $\gf_{2^n}$. Then the following
forms a complete set of MUBs:
$$B_{\infty}=\{e_{w}| w \in \gf_{2^n}\}, \ \ B_{m}=\{b_{m, v}| v \in \gf_{2^n}\}, m \in \gf_{2^n}, $$
$$b_{m, v} = \frac{1}{\sqrt{2^n}} \sum_{w\in \gf_{2^n}} \omega^{\tr_R\left(\widehat{m}(\widehat{w}^2 + 2F(\widehat{w})) + 2 \widehat{v}\widehat{w} \right)} e_{w}, $$
where $B_{\infty}=\{e_{w}| w \in \gf_{2^n}\}$ is the standard basis of the $2^n$-dimensional Hilbert space,
$\omega=\sqrt{-1}$ is the primitive $4$-root of unity, and $\widehat{m}$ is the Teichm$\ddot{\text{u}}$ller lift of $m$.
\end{Th}

Since $\{B_{\infty}, B_{m}, m \in \gf_{2^n}\}$ forms a complete set of MUB, the square of the
absolute value of the inner product $|(x, y)|^2$ is equal to $1/2^n$ for any two vectors $x,
y$ from distinct bases. Then the following result follows directly from \eqref{eq_Lev_boundC},
which give  explicit expression  of  the codebook in Corollary \ref{Cor_codebook}.
\begin{Prop}\label{prop_codebook}
Let $F$, $B_{\infty}$ and $B_{m}$ be defined as in Theorem \ref{th_Abdukhalikov}, and let
$C=B_{\infty}\cup B_{m}$. Then $\mathcal{C}$ is an optimal $(2^{2n} + 2^n, 2^n)$  complex codebook meeting
Levenstein bound with alphabet size $6$.
\end{Prop}

As pointed out by Zhou, Ding and Li \cite{ZDL14},
codebooks achieving the Levenstein bound  can be used in compressed sensing.
Compressed sensing is a novel sampling theory, which provides
a fundamentally new approach to data acquisition. A central
problem in compressed sensing is the construction of the sensing
matrix. For more information on the theory of compressed
sensing, the reader is referred to Donoho \cite{Donoho} and Cand\`{e}s and
Tao \cite{CandesTao}. Recently, Li, Gao,  Ge et. al. \cite{LGGZ14}  found that codebooks
achieving the Levenstein bound can be used to construct deterministic
sensing matrices with smallest coherence. The numerical
experiments conducted in \cite{LGGZ14} showed that the sensing
matrices from some known codebooks meeting the Levenstein
bound have a good performance. Since a pseudo-planar function leads to an optimal codebook
meeting the Levenstein bound, it would be interesting to
investigate the application of these codebooks constructed in this paper using
the framework developed in \cite{LGGZ14}.

Hence a pseudo-planar function over $\gf_{2^n}$ not only gives rise to a finite projective plane
and a relative difference set, it also leads to a complete set of MUB in $\mathbb{C}^{2^n}$, an optimal $(2^{2n} + 2^n, 2^n)$  complex codebook meeting
the Levenstein bound, and compressed sensing matrices with low coherence. These interesting links are
the motivations for the author to study the construction of pseudo-planar functions.

\subsection{Other Results}

In this subsection, we review some necessary definitions and results for future use.
For a nonzero element $\alpha$ in $\gf_{2^n}$, $\text{Ord}(\alpha)$ denotes the multiplicative
order of $\alpha$, that is, the smallest positive integer $t$ such that $\alpha^t=1$.
Let $k$ be a divisor of $n$. Then for $\alpha\in \gf_{2^n}$, the trace $\tr_{n/k}(\alpha)$ of $\alpha$ over $\gf_{2^k}$ is defined by
$$\tr_{n/k}(\alpha)=\alpha + \alpha^{2^k} + \alpha^{2^{2k}} +\cdots + \alpha^{2^{n-k}},$$
the norm  $\text{N}_{n/k}(\alpha)$ of $\alpha$ over $\gf_{2^k}$ is defined by
$$\text{N}_{n/k}(\alpha)=\alpha \cdot \alpha^{2^k} \cdot \alpha^{2^{2k}} \cdot \cdots \cdot \alpha^{2^{n-k}} = \alpha^{\frac{2^n-1}{2^k-1}}.$$
\begin{Lemma}\cite{Lidl}
  \label{redu-condition}
  For any $a,b\in\gf_{2^n}$ and $a\ne 0$, the polynomial $p(x)=x^2+ax+b\in\gf_{2^n}[x]$ is irreducible
  if and only if $\tr_{n/1}(b/a^2)=1$.
\end{Lemma}

\begin{Lemma}\cite[Theorem 7.7]{Lidl}
\label{le_PP_dual}
A mapping $f : \gf_{2^n}\rightarrow \gf_{2^n}$ is a permutation polynomial of $\gf_{2^n}$
if and only if for every nonzero $b\in \gf_{2^n}$,
\begin{equation*}
    \sum\limits_{x\in \gf_{2^n}} (-1)^{\tr_{n/1} (bf(x))} = 0.
\end{equation*}
\end{Lemma}

\begin{Lemma}\cite[P. 362]{Lidl}
\label{le_PP_det}
Let $q$ be a prime power and $\gf_{q^t}$ be an extension of $\gf_q$. Then the
linearized polynomial
$$L(x)=\sum_{i=0}^{t-1} a_i x^{q^i}\in \gf_{q^t}[x]$$
 is a permutation polynomial of $\gf_{q^t}$
if and only if the Dickson determinant of $a_0, a_1, \cdots, a_{t-1}$ is nonzero,
that is,
\begin{equation*}
    \det \left(\begin{array}{ccccc}
                       a_0 & a_1 & a_2 & \cdots & a_{t-1} \\
                       a_{t-1}^q & a_0^q & a_1^q  & \cdots & a_{t-2}^q \\
                      \vdots &  \vdots &  \vdots & &   \vdots\\
                       a_1^{q^{t-1}} & a_2^{q^{t-1}} & a_3^{q^{t-1}} & \cdots  &  a_0^{q^{t-1}}
                     \end{array}
  \right) \neq 0.
\end{equation*}

\end{Lemma}
%

%
%
%
%

\section{A New Approach to Constructing Quadratic pseudo-planar Functions over $\gf_{2^n}$}
%
%

\subsection{A General Family of Quadratic pseudo-planar Functions }

\begin{Th}\label{th_Gen}
Assume $n=tm (t\geq 2)$ and $q=2^m$. Let
\begin{equation}\label{eq_gen_Form}
    \begin{array}{ccc}
  F(x) &=& \sum\limits_{i=0}^{(t-1)m-1} c_{1, i}x^{2^i(q+1)} + \sum\limits_{i=0}^{(t-2)m-1} c_{2, i}x^{2^i(q^2+1)}  \\
   && \ \ \ \ + \cdots +
\sum\limits_{i=0}^{m-1} c_{t-1, i}x^{2^i(q^{t-1}+1)}\in \gf_{2^n}[x].
    \end{array}
\end{equation}
%
Then $F$ is pseudo-planar over $\gf_{2^n}$ if and only if
\begin{equation}\label{eq_detMb_Gen}
\det M_b =\left|\begin{array}{ccccc}
                       A_0 & A_1 & A_2 & \cdots & A_{t-1} \\
                       A_{t-1}^q & A_0^q & A_1^q  & \cdots & A_{t-2}^q \\
                      \vdots &  \vdots &  \vdots & &   \vdots\\
                       A_1^{q^{t-1}} & A_2^{q^{t-1}} & A_3^{q^{t-1}} & \cdots  &  A_0^{q^{t-1}}
                     \end{array}
  \right| \neq   0
\end{equation}
for any nonzero $b$ in $\gf_{2^n}$, where
\begin{equation}\label{eq_Aall}
    \left\{\begin{array}{lll}
             A_0 & = & b, \\
             A_1 &=&  \sum\limits_{i=0}^{(t-1)m-1}\left(c_{1,i}b\right)^{2^{n-i}} + \sum\limits_{i=0}^{m-1}\left(c_{t-1,i}b\right)^{2^{m-i}},  \\
             &\vdots&\\
             A_j & = & \sum\limits_{i=0}^{(t-j)m-1}\left(c_{j,i}b\right)^{2^{n-i}} + \sum\limits_{i=0}^{jm-1}\left(c_{t-j,i}b\right)^{2^{jm-i}}, \\
             &\vdots&\\
            A_{t-1} & = & \sum\limits_{i=0}^{m-1}\left(c_{t-1,i}b\right)^{2^{n-i}} + \sum\limits_{i=0}^{(t-1)m-1}\left(c_{1,i}b\right)^{2^{(t-1)m-i}}.
           \end{array}
\right.
\end{equation}
Moreover, we have
\begin{equation}\label{eq_Atj}
    A_j=A_{t-j}^{q^j}, \ \text{for all}\  1\leq j\leq t-1.
\end{equation}

\end{Th}

\begin{proof}
We only prove the first part. The second part can be verified directly from \eqref{eq_Aall}, that is,  the definitions of
$A_i, \ 0\leq i\leq t-1$.

It is clear that $F$ is pseudo-planar if and only if
\begin{eqnarray*}
   \mathbb{L}_a(x) &:=& F(x+a) + F(x) + F(a) + ax \\
   &=& \sum\limits_{i=0}^{(t-1)m-1} c_{1, i}\left(a^{2^i}x^{2^{m+i}} + a^{2^{m+i}}x^{2^i}\right) \\
    && \ + \sum\limits_{i=0}^{(t-2)m-1} c_{2, i}\left(a^{2^i}x^{2^{2m+i}} + a^{2^{2m+i}}x^{2^i}\right) \\
    && \ + \cdots \\
   & & \ + \sum\limits_{i=0}^{m-1} c_{t-1, i}\left(a^{2^i}x^{2^{(t-1)m+i}} + a^{2^{(t-1)m+i}}x^{2^i}\right)\\
   & & \  + ax
\end{eqnarray*}
is a linearized permutation polynomial over $\gf_{2^n}$ for any nonzero $a$ in $\gf_{2^n}$, or equivalently,
$\mathbb{L}_a(x)=0$ if and only if $x=0$ or $a=0$.

Instead of investigating $\mathbb{L}_a(x)$ directly, we turn to  studying  its dual
linearized polynomial. Thanks to the character theory, we can do this transformation as follows.

According to Lemma \ref{le_PP_dual},
$\mathbb{L}_a(x)$ is a linearized permutation polynomial over $\gf_{2^n}$ for any nonzero $a$ in $\gf_{2^n}$
if and only if for every nonzero $b\in \gf_{2^n}$,
\begin{equation*}
    0 = \sum\limits_{x\in \gf_{2^n}} (-1)^{\tr_{n/1} (b \mathbb{L}_a(x))} = \sum\limits_{x\in \gf_{2^n}} (-1)^{\tr_{n/1} (\mathbb{L}^\ast_b(a)  x )},
\end{equation*}
and if and only if
\begin{equation*}
    \mathbb{L}^\ast_b(a) \neq 0, \ \  \text{for all}\  a, b \in \gf_{2^n}^\ast,
\end{equation*}
where
\begin{eqnarray*}
   && \mathbb{L}^\ast_b(a)\\
  &=& \sum\limits_{i=0}^{(t-1)m-1} \left( (c_{1,i}a^{2^i}b)^{2^{(t-1)m-i}} + (c_{1, i}a^{2^{m+i}}b)^{2^{n-i}} \right) \\
   & & \ \  + \sum\limits_{i=0}^{(t-2)m-1} \left( (c_{2,i}a^{2^i}b)^{2^{(t-2)m-i}} + (c_{2,i}a^{2^{2m+i}}b)^{2^{n-i}}  \right) \\
   & & \ \ + \cdots \\
    & & \ \  + \sum\limits_{i=0}^{m-1} \left( (c_{t-1,i}a^{2^i}b)^{2^{m-i}} + (c_{t-1,i}a^{2^{(t-1)m+i}}b)^{2^{n-i}}  \right)  \\
    && \ \ + a b.
\end{eqnarray*}

Hence $F$ is pseudo-planar if and only if $\mathbb{L}^\ast_b(a)$
is a linearized permutation polynomial for any nonzero $b\in \gf_{2^n}$.

Then the result follows directly  from Lemma \ref{le_PP_det} and
\begin{eqnarray*}
   \mathbb{L}^\ast_b(a)&=& A_0\cdot a +  A_1\cdot a^{2^m}  + \cdots + A_{t-1}\cdot a^{2^{(t-1)m}},
\end{eqnarray*}
where $A_0, A_1, \cdots, A_{t-1}$ are defined in \eqref{eq_Aall}.
\end{proof}

A general family of quadratic pseudo-planar functions is constructed by Theorem \ref{th_Gen}.
Given a quadratic function $F$ in this family, a sufficient and necessary condition
for it to be pseudo-planar is presented. This condition is deduced from the permutation property
of the dual polynomial
$\mathbb{L}^\ast_b(a)$ of the corresponding derivative polynomial
$\mathbb{L}_a(x)$. It seems that this condition have
additional properties and it is more easily handled than the condition
deduced directly from the permutation property of $\mathbb{L}_a(x)$.
Combining this benefit with the technique that will be introduced in the next subsection,
we can construct several  families of pseudo-planar functions with new explicit forms,
reconstruct and generalize known families.

In the end of this subsection, we would like to point out that the function in \eqref{eq_Kantor_Pla},
that is, the pesudo-planar function from the semifields of the Kantor family, is with the form \eqref{eq_gen_Form}.
To see this, let $\gf_i=\gf_{2^{t_im}}$, $0\leq i\leq r$, where $1=t_r |t_{r-1}|\cdots | t_1| t_0=t $ and
$t$ is odd. Then it is clear that the function in \eqref{eq_Kantor_Pla} is with the form \eqref{eq_gen_Form}.
Hence all the known pesudo-planar functions are included in the general family of functions constructed by
Theorem \ref{th_Gen}.

\subsection{Discussing $\det M_b$}

According to Theorem \ref{th_Gen}, to discuss the pseudo-planarity of $F$ with the form
of \eqref{eq_gen_Form},
we need to discuss whether $\det M_b\neq 0$ or not,
where $\det M_b$ is defined by \eqref{eq_detMb_Gen}.
We will introduce a technique. It is generalized from a trick
which was firstly  used in the proof of \cite[Theorem 3.1]{DXY} and then in
the proof of \cite[Proposition 3.6]{HuLiZhang}.
Let us set up the following notations.

Throughout this subsection, let $q=2^m$ and $n=tm$, where $t\geq 2$.
For a nonzero $b$ in $\gf_{2^n}$, we define
$$x_1=b, \ x_2=b^q, \ \cdots, x_t=b^{q^{t-1}},$$
and let $B_1, B_2, \cdots, B_t$ be the first $t$ elementary symmetric
polynomial with variables $x_1, x_2, \cdots, x_t$, that is

 \begin{equation}\label{eq_B1Bt}
\left\{\begin{array}{lll}
   B_1 &=& x_1 + x_2 + \cdots + x_t = \tr_{n/m}(b), \\
    B_2 &=& \sum\limits_{1\leq i<j\leq t}x_i x_j,  \\
    &\vdots& \\
     B_t &=& x_1  x_2 \cdots x_t = {{\rm{N}}}_{n/m}(b).
     \end{array}\right.
\end{equation}
Denote the characteristic polynomial of $b$ over $\gf_q$ by
$$m_b(x)=(x+b)(x+b^q)\cdots(x+b^{q^{t-1}}).$$
Then we have
\begin{equation}\label{eq_mb_Gen}
m_b(x) = x^t + B_1 x^{t-1} + \cdots + B_{t-1}x + B_t\in \gf_q[x].
\end{equation}

It is clear that $m_b(x)$ is irreducible over $\gf_q$ if and
only if $b$ is not in any proper subfield of $\gf_{q^t}$.

Since $\det M_b$ is a Dickson determinant of $A_0, A_1, \cdots$, $A_{t-1}$,
where each $A_i$ is a linearized polynomial of $b$,
$\det M_b$ can be regarded as a homogenous multi-polynomial
of $x_1, x_2, \cdots, x_t$ with degree $t$.
If $b$ is in some proper subfield $\gf_{q^r}$ of $\gf_{q^t}$, then $\det M_b$ can be simplified since
$x_1, x_2, \cdots, x_t$ are just $t/r$ repetitions of $x_1, x_2, \cdots, x_r$.
Hence it is usually easy to discuss whether $\det M_b\neq 0$ or not.
We assume that $\det M_b\neq 0$ always holds in this case. Otherwise, $F$ can not
be a pseudo-planar function. In the following, we assume that $b$ is not in any proper subfield of $\gf_{q^t}$.
Then $m_b(x)$ is an irreducible polynomial over $\gf_q$. We distinguish two cases
according to whether $\det M_b$ is symmetric over $x_1, x_2, \cdots, x_t$
or not.

{\bf Case 1: $\det M_b$ is symmetric.}

Since $\det M_b$ is symmetric over $x_1, x_2, \cdots, x_t$, it follows from the theory
of linear algebra that  $\det M_b$ can be
expressed as a polynomial of $B_1, B_2, \cdots, B_t$, the first $t$
elementary symmetric polynomial of $x_1, x_2, \cdots, x_t$.
Then the assumption $\det M_b=0$ is equivalent to a relation, called \emph{Relation X} for convenience,  between $B_1, B_2, \cdots, B_t$.
If $m_b(x)$  is reducible over $\gf_q$
for any  $B_1, B_2, \cdots, B_t$ satisfying \emph{Relation X}, then this contradicts the assumption that  $m_b(x)$ is irreducible over
$\gf_q$, which means that $\det M_b=0$ is impossible for any nonzero $b$. Hence $F$ is pseudo-planar.
On the other hand, if there exists a collection of $B_1, B_2, \cdots, B_t$ satisfying \emph{Relation X} such that
$m_b(x)$,  defined by \eqref{eq_mb_Gen},  is irreducible over $\gf_q$, then
a zero of $m_b(x)$, denoted by $\beta$, will satisfy $\det M_{\beta}=0$, which means that
$F$ is not pseudo-planar.
 Thus the problem of checking the pseudo-planarity of $F$ is converted to discussing whether there exists an irreducible polynomial
$m_b(x)$ (defined by \eqref{eq_mb_Gen}) such that its coefficients $B_1, B_2, \cdots, B_t$ satisfy \emph{Relation X}.
This discussion may split into two subcases according to whether $B_1=0$ or not.
For more details, we refer the readers to the proofs in the next section.

{\bf Case 2: $\det M_b$ is not symmetric.}

Then $\det M_b$ can be expressed as the summation of its symmetric part over $x_1, x_2, \cdots, x_t$, denoted by $s$,
and its non-symmetric part, denoted by $t_1$. It is clear that $s$ can be  expressed as a polynomial of $B_1, B_2, \cdots, B_t$.
For the non-symmetric part $t_1$, let  $t_2, \cdots, t_k$ be the distinct
images of $t_1$ under all the  permutation transformations of  $x_1, x_2, \cdots, x_t$ (cf. $t_2$ in the proof of Theorem \ref{th_ind3Class2},
and $t_2, t_3$ in the proof of Theorem \ref{th_ind4Class1}).
Then all the first  $k$ elementary symmetric polynomials of $t_1, t_2, \cdots, t_k$ can
be expressed as a polynomial of $B_1, B_2, \cdots, B_t$ since they are also symmetric over $x_1, x_2, \cdots, x_t$.
Hence we get $k$ relations between $t_1, t_2, \cdots, t_k$ and $B_1, B_2, \cdots, B_t$.

Now assume that $\det M_b=0$. Then $t_1$ can be expressed by $B_1, B_2, \cdots, B_t$. Substituting it into the aforementioned
$k$ relations, one may get a relation between $B_1, B_2, \cdots, B_t$ as in Case 1, even though
this relation is usually much complicated. Similarly, if for any collection of $B_1, B_2, \cdots, B_t$
satisfying the aforementioned relation, $m_b(x)$ can  be proved to be reducible over $\gf_q$,
or $\det M_b\neq 0$ holds for any zero of the irreducible polynomial $m_b(x)$,
then $F$ is pseudo-planar.

\section{Families of Quadratic pseudo-planar Functions with New Explicit Forms}

In this section, we will use the new approach introduced in the last
section to construct
several families of quadratic pseudo-planar functions  with new explicit forms over $\gf_{2^n}$,
and reconstruct known families.
The section is divided into three subsections according to the values of $t$, the extension degree of $\gf_{2^n}$  over $\gf_{2^m}$.
We begin with the case of $t=3$.
We construct three new families of pseudo-planar functions, and study a family of trinomials, which
is a generalization of the three families of functions in \cite{HuLiZhang}.
The monomial polynomial is also revisited, and a sufficient and necessary condition for it to be pseudo-planar is given.
For the extension degree $4$ case, we construct two new families of pseudo-planar functions.
One is a trinomial, the other is a quadrinomial.
For the extension degree $2$ case, we revisit the monomial pseudo-planar function and provide a simple
sufficient and necessary condition, which generalizes \cite[Theorem 6]{SchmidtZhou}. However,
we cannot construct new pseudo-planar function in this case and leave it as an open problem.

\subsection{Case 1: Extension Degree $t=3$}

\begin{Th}\label{th_ind3Gen}
Set $n=3m$ and $q=2^m$.  Let
$$F(x)=\sum\limits_{i=0}^{2m-1} c_{1,i}x^{2^{m+i}+2^i} + \sum\limits_{i=0}^{m-1} c_{2,i}x^{2^{2m+i}+2^i} \in \gf_{2^n}[x].$$
Then $F$ is pseudo-planar over $\gf_{2^n}$ if and only if
$$b^{q^2+q+1} + \tr_{n/m}(b^q A_2^2)\neq 0$$
for any nonzero $b$ in $\gf_{2^n}$, where
$$A_2 = \sum\limits_{i=0}^{m-1}(c_{2,i}b)^{2^{n-i}} + \sum\limits_{i=0}^{2m-1}(c_{1,i}b)^{2^{2m-i}}.$$
\end{Th}
\begin{proof}
%
According to Theorem \ref{th_Gen}, the dual linearized polynomial of $ \mathbb{L}_a(x) = F(x+a) + F(x) + F(a) + ax$ is
$\mathbb{L}^\ast_b(a)$:

\begin{eqnarray*}
  \mathbb{L}^\ast_b(a) &=& A_0\cdot a +  A_1\cdot a^{2^m} + A_2\cdot a^{2^{2m}},
\end{eqnarray*}
where

\begin{equation*}
    \left\{\begin{array}{lll}
             A_0 & = & b, \\
             A_1 &=&  \sum\limits_{i=0}^{2m-1} (c_{1,i}b)^{2^{n-i}}+ \sum\limits_{i=0}^{m-1} (c_{2,i}b)^{2^{m-i}} = A_2^q,  \\
             A_2 & = & \sum\limits_{i=0}^{m-1}(c_{2,i}b)^{2^{n-i}} + \sum\limits_{i=0}^{2m-1}(c_{1,i}b)^{2^{2m-i}}.
           \end{array}
\right.
\end{equation*}
Then
\begin{eqnarray*}
  \det M_b &=& \left|\begin{array}{ccc}
                        A_0 & A_1 & A_2 \\
                       A_2^q & A_0^q & A_1^q \\
                       A_1^{q^2} & A_2^{q^2} & A_0^{q^2}
                     \end{array}
  \right| \\
&= & A_0^{q^2+q+1} + A_1^{q^2+q+1} + A_2^{q^2+q+1} \\
&& + \tr_{n/m}\left(A_0A_1^qA_2^{q^2}\right)\\
  &=& b^{q^2+q+1} + \tr_{n/m}\left(bA_2^{2q^2}\right)\\
  & =&   b^{q^2+q+1} + \tr_{n/m}\left(b^qA_2^2\right).
\end{eqnarray*}
Hence the result follows directly from Theorem \ref{th_Gen}.
\end{proof}

\begin{Th}\label{th_ind3Class1}
Set $q=2^m$ and $n=3m$.  Let
$$F(x)=c x^{2(q+1)} + c^q x^{2(q^2+1)} \in \gf_{2^n}[x].$$
Then $F$ is pseudo-planar over $\gf_{2^n}$.
\end{Th}
\begin{proof}
By Theorem \ref{th_ind3Gen},  we have
$$A_2=(c^qb)^{2^{n-1}} + (c b)^{2^{2m-1}}.$$
Then it follows that
$$\tr_{n/m}\left(b^qA_2^2\right) = \tr_{n/m}\left(c^qb^{q+1} + c^{q^2}b^{q^2+q} \right) \equiv 0. $$
Hence
$$\det M_b = b^{q^2+q+1} + \tr_{n/m}\left(b^qA_2^2\right)= b^{q^2+q+1}\neq 0$$
for any nonzero $b$ in $\gf_{2^n}$. Then the result follows directly from Theorem \ref{th_ind3Gen}.
\end{proof}

Before introducing the second family of pseudo-planar function, we
set up some notations as in Section III.B.
Let
$$x_1=b, x_2=b^q, \text{\  and \ } x_3=b^{q^2}.$$
Then \eqref{eq_B1Bt} and \eqref{eq_mb_Gen} become
 \begin{eqnarray*}
   B_1 &=& x_1 + x_2 + x_3 = \tr_{n/m}(b), \\
    B_2 &=& x_1x_2 + x_1x_3 + x_2x_3,  \\
     B_3 &=& x_1  x_2  x_3 = {\rm{N}}_{n/m}(b),
\end{eqnarray*}
and
$$m_b(x)=x^3 + B_1x^2 + B_2x + B_3\in \gf_q[x].$$
The following identity can be easily verified.
\begin{equation}\label{eq_ind3_Trb3}
    \tr_{n/m}(b^3) = x_1^3 + x_2^3 + x_3^3 = B_1^3 + B_3 + B_1B_2.
\end{equation}

\begin{Th}\label{th_ind3Class2}
Set $q=2^m$ and $n=3m$.  Let
$$F(x)= x^{2(q+1)} + x^{q^2+1} + x^{q^2+q} +  x^{2(q^2+1)}.$$
Then $F$ is pseudo-planar over $\gf_{2^n}$ if and only if $m\not\equiv 1\bmod 3$.
\end{Th}
\begin{proof}
According to Theorem \ref{th_ind3Gen}, we have
$$A_2=b^{2^{2m-1}} + b^{2^{m}} + b + b^{2^{n-1}}.$$
Then it follows from Theorem \ref{th_ind3Gen} that
\begin{eqnarray*}
  && \det M_b\\
   &=& b^{q^2+q+1} + \tr_{n/m} (b^qA_2^2)\\
   &=&  b^{q^2+q+1} + \tr_{n/m} \left(b^{q^2+q} + b^{3q} + b^{q+2} + b^{q+1}\right)\\
   &=&  b^{q^2+q+1} + \tr_{n/m} \left(b^{3} + b^{q+2}\right).
\end{eqnarray*}

Then with \eqref{eq_ind3_Trb3}, we have
\begin{equation}\label{eq_ind3Mb}
    \det M_b = B_1^3 + B_1B_2 + t_1,
\end{equation}
where
\begin{equation}\label{eq_ind3t1}
    t_1 = \tr_{n/m}( b^{q+2}) = x_1^2x_2 + x_2^2x_3 + x_3^2x_1.
\end{equation}
Let $t_2$ be the image of $t_1$ under the transformation of
$(12)$, that is, to exchange $x_1$ and $x_2$.
 \begin{eqnarray*}
   t_2 &=& x_1x_2^2 + x_2x_3^2 + x_3x_1^2.\\
\end{eqnarray*}
Then the following identities hold.
\begin{equation}\label{eq_ind3_tsum}
    t_1 + t_2=  B_3 + B_1B_2,
\end{equation}
\begin{equation}\label{eq_ind3_tquad}
    t_1t_2 =   B_1^3B_3 + B_2^3+ B_3^2.
\end{equation}

Firstly, we assume that $b\in \gf_{q}^\ast$. Then we have $B_1=b$,  $B_2=b^2$ and $t_1=b^3$.
Hence
$$\det M_b = b^3\neq 0$$
for any $b\in \gf_{q}^\ast$. In the following, we always assume that $b\in \gf_{q^3}^\ast \setminus \gf_{q}$,
which  means that $m_b(x)$ is an irreducible polynomial over $\gf_q$ with degree $3$.

Let $\gamma$ be a solution of $y^3+y+1=0$ in some extension field of $\gf_q$.
Then $\text{Ord}(\gamma)=7$.

If $m\equiv 1\bmod 3$, then $q \equiv 2\bmod 7$. Further, we have
\begin{eqnarray*}
   \det M_\gamma&=& \gamma^{q^2+q+1} + \tr_{n/m} \left(\gamma^{3} + \gamma^{q+2}\right)\\
   &=& 1+\tr_{n/m}(\gamma^3+\gamma^4)=1+\tr_{n/m}(\gamma^6)\\
   &=& 1 + \gamma^6 +  \gamma^5 + \gamma^3 =0,
\end{eqnarray*}
which means that $F$ is not pseudo-planar.
In the rest of the proof, we always assume that $m\not\equiv 1 \bmod 3$.
It suffices to prove that $\det M_b\neq 0$ for any $b\in \gf_{q^3}^\ast \setminus \gf_q$.

The following proof is split into two cases according to $B_1=0$ or not.

{\bf Case 1:  $B_1=0$.}

Now \eqref{eq_ind3Mb} becomes
\begin{equation}\label{eq_ind3_B10_Mb}
\det M_b = t_1.
\end{equation}

Assume that $\det M_b = t_1 = 0$ for some $b\in \gf_{q^3}^\ast \setminus \gf_{q}$.
Plugging it with $B_1=0$ into \eqref{eq_ind3_tquad}, one gets
\begin{equation}\label{eq_B10_t23S}
    B_3 = B_2^{3/2}.
\end{equation}

Then it follows that $B_2\neq 0$ since $B_3\neq 0$.
Let $x=B_2^{1/2}y$. Then we have
\begin{eqnarray*}
  m_b(x)   &=&  x^3 + B_2x + B_3  =   x^3 + B_2x + B_2^{3/2}\\
  &=& B_2^{3/2}(y^3+y+1).
\end{eqnarray*}
Hence $$b\in \{B_2^{1/2}\gamma, B_2^{1/2}\gamma^2, B_2^{1/2}\gamma^4\}.$$

If $m\equiv 0\bmod 3$, then both $\gamma$ and $b$ are in $\gf_q$. Contradicts!

If $m\equiv 2\bmod 3$, then $q\equiv 4\bmod 7$ and
\begin{eqnarray*}
  && \det M_{B_2^{1/2}\gamma} \\
   &=& t_1   = \tr_{n/m}(B_2^{3/2}\gamma^{q+2})=B_2^{3/2}\tr_{n/m}(\gamma^6)\\
    &=& B_2^{3/2}(\gamma^6 + \gamma^3 + \gamma^5) = B_2^{3/2}\neq 0,
\end{eqnarray*}
which is also a contradiction.

Hence $\det M_b\neq 0$ if $B_1=0$.

{\bf Case 2:  $B_1\neq 0$.}

It is clear that $\det M_{cb}=c^3 \det M_b$ holds for any $c\in \gf_q^\ast$.
Hence, WLOG, we assume that $B_1=1$.
Then  \eqref{eq_ind3Mb} becomes
\begin{equation}\label{eq_ind3_B11_Mb}
\det M_b = B_2 + t_1 + 1.
\end{equation}

Assume, on the contrary,  that $\det M_b = 0$ for some $b\in \gf_{q^3}^\ast \setminus \gf_{q}$.
Then it follows from \eqref{eq_ind3_B11_Mb} that
\begin{equation}\label{eq_ind3B11_t1}
    t_1 = B_2 + 1.
\end{equation}

Plugging it with $B_1=1$ into \eqref{eq_ind3_tsum}, one gets
\begin{equation}\label{eq_ind3B11_t2}
    t_2 = B_3 +1.
\end{equation}

Substituting \eqref{eq_ind3B11_t1}, \eqref{eq_ind3B11_t2}  and  $B_1=1$ into
\eqref{eq_ind3_tquad}, we have
\begin{equation}\label{eq_ind3_B11}
    B_3^2 + B_2B_3 + B_2^3 + B_2 + 1 =0.
\end{equation}

 We distinguish two subcases.

 {\bf Subcase 2.1: $B_2 = 0$. }

 Then it follows from \eqref{eq_ind3_B11} that $B_3=1$.
 Hence
 $$m_b(x) = x^3+x^2+1,$$
 which implies that $$b\in \{\gamma^3, \gamma^6,  \gamma^5\}.$$
 A similar argument as in the last case can show that $\det M_b\neq 0$.

 {\bf Subcase 2.2: $B_2\neq 0$. }

Let $u=\frac{B_3+1}{B_2}\in \gf_q$. Then dividing $B_2^2$ across both sides of \eqref{eq_ind3_B11} leads
to
$$B_2=u^2+u.$$
Further,
$$B_3=uB_2+1=u^3+u^2+1.$$
We compute that
\begin{eqnarray*}
   && (\gamma u + \gamma^6  )^3 + (\gamma u + \gamma^6  )^2  + B_2(\gamma u + \gamma^6  ) + B_3 \\
      &=&  (\gamma^3 + \gamma + 1) u^3 + (\gamma^6 + \gamma^2 + 1) u^2 + (\gamma^4 + \gamma^5 + 1)\\
      & =& 0.
\end{eqnarray*}
Hence $b_0=\gamma u + \gamma^6  $ is a zero of $m_b(x) = x^3 + x^2 + B_2x+B_3 $.

If $m\equiv 0\bmod 3$, then $b_0\in \gf_{q}$, which contradicts that $m_b(x)$ is irreducible.
 If $m\equiv 2\bmod 3$, then $q\equiv 4\bmod 7$ and
 a direct computation shows that
 $$B_2 = b_0^{q+1} + b_0^{q^2+1} + b_0^{q^2+q} = u^2+u$$
 and
 \begin{eqnarray*}
   t_1 &=& \tr_{n/m}(b_0^{q+2}) \\
   &=& \tr_{n/m}\left(\gamma^6 u^3+ \gamma^5 u^2 + \gamma^2 u + \gamma\right)\\
   & =&  u^3 + u^2.
 \end{eqnarray*}
Hence
$$\det M_{b_0}=B_2+t_1+1=u^3+u+1\neq 0$$
since $u\in \gf_q$ and $\gcd(7, q-1)=1$.
Contradicts!

We finish the proof.
\end{proof}

\begin{Prop}\label{prop_ind3Trinomial}
Set $q=2^m$ and $n=3m$.  Let
\begin{equation}\label{eq_ind3Tri}
    F(x)= c_1x^{q+1} + c_2x^{q^2+q} + c_3x^{q^2+1}.
\end{equation}
Then $F$ is pseudo-planar over $\gf_{2^n}$ if and only if
\begin{equation}\label{eq_ind3TriCond}
    b^{q^2+q+1} + \tr_{n/m}\left(c_1^{2}b^{q^2+2} + c_2^{2}b^{3} + c_3^2b^{q+2}\right)\neq 0
\end{equation}
for any $b\in \gf_{2^n}^\ast$.
\end{Prop}
\begin{proof}
In this case, we have
$$A_2=(c_1 b)^{2^{2m}} + (c_2 b)^{2^{m}} + c_3 b.$$
Then the result  follows from Theorem \ref{th_ind3Gen} and
\begin{eqnarray*}
  &&\tr_{n/m}\left(b^qA_2^2\right)\\ &=& \tr_{n/m}\left(c_1^{2q^2}b^{2q^2+q} + c_2^{2q}b^{3q} + c_3^2b^{q+2}
\right) \\
&=&  \tr_{n/m}\left(c_1^{2}b^{q^2+2} + c_2^{2}b^{3} + c_3^2b^{q+2}
\right).
\end{eqnarray*}

\end{proof}

Experiment results show that there are a lot of pseudo-planar functions
with the form \eqref{eq_ind3Tri}. We use Magma to do an exhaustive search
over $\gf_{2^{3m}}$ for $m=1, 2, 3$. Results show that
there are $8$, $960$ and $75264$ pseudo-planar functions
with the form \eqref{eq_ind3Tri} over $\gf_{2^{3}}$,  $\gf_{2^{6}}$
and  $\gf_{2^9}$ respectively.

\begin{Cor}\label{cor_ind3Class3}
Set $q=2^m$ and $n=3m$.  Let
$$F(x)=x^{q+1} + \alpha x^{q^2+q} + x^{q^2+1}, $$
where $\alpha$ is a solution of $x^3+x^2+1=0$.
Then $F$ is pseudo-planar over $\gf_{2^n}$.
\end{Cor}

\begin{proof}
Clearly such $\alpha$ does exist in $\gf_{2^3}^\ast$.
According to Proposition \ref{prop_ind3Trinomial}, we have
\begin{eqnarray*}
  &&\det M_b  \\
  &=& b^{q^2+q+1} + \tr_{n/m}\left(b^{q^2+2} + \alpha^{2}b^{3} + b^{q+2}\right) \\
  &=& B_3 + (B_3+ B_1B_2) + \alpha^2(B_1^3 + B_3 + B_1B_2) \\
   &=& \alpha^2B_3 + \alpha^2 B_1^3 + (1+\alpha^2)B_1B_2.
\end{eqnarray*}
Then a similar but  much simple argument as in Theorem \ref{th_ind3Class1} will prove this corollary.
We leave it to the interested readers.
\end{proof}

Several classes of known constructions can be explained by Proposition \ref{prop_ind3Trinomial}.
\begin{example}\label{ex_known_Pla}
In Proposition \ref{prop_ind3Trinomial},

(1) Let $c_1=0$ and $c_2=c_3=1$. Then $F(x)= x^{q^2+q} + x^{q^2+1}$ is pseudo-planar over $\gf_{2^n}$ if and only if
$$b^{q^2+q+1} + \tr_{n/m}\left(b^{q+2} + b^{3} \right)\neq 0$$
for any $b\in \gf_{2^n}^\ast$, which is the same equation as in Theorem \ref{th_ind3Class1}.
Hence $F$ is pseudo-planar if and only if $m\not\equiv 1\mod 3$. This is  \cite[Proposition 3.8]{HuLiZhang}.

(2) Let $c_1=c_2=1$ and $c_3=0$. Then $F(x)= x^{q+1} + x^{q^2+q}$ is pseudo-planar over $\gf_{2^n}$ if and only if
$$b^{q^2+q+1} + \tr_{n/m}\left(b^{q^2+2} + b^{3} \right)\neq 0$$
for any $b\in \gf_{2^n}^\ast$, which holds if and only if $m\not\equiv 2\mod 3$ by
a similar proof as in Theorem \ref{th_ind3Class1}. This is \cite[Proposition 3.6]{HuLiZhang}.

(3) Let $c_1=a^{-(q+1)}$, $c_2=0$ and $c_3=a^{q^2+1}$. Then $F(x)= a^{-(q+1)}x^{q+1} + a^{q^2+1}x^{q^2+1}$ is pseudo-planar over $\gf_{2^n}$ if and only if
\begin{equation}\label{eq_new_condition}
    b^{q^2+q+1} + \tr_{n/m}\left( a^{-2(q^2+q)}b^{2q+1}  + a^{2(q^2+1)}b^{q+2} \right)\neq 0
\end{equation}
for all $b\in \gf_{2^n}^\ast$. In   \cite[Proposition 3.2]{HuLiZhang}, a sufficient and necessary condition for
$F$ to be pseudo-planar was given as follows.
\begin{equation}\label{eq_HLZ_condition}
\begin{array}{lll}
  &&\tr_{n/m}\left( (a^{q^2+q} + a^{-q^2-q-2})(a^{q+1} + b^{q-1}) b^{q+2} \right.\\
  && \ \ \ \left. + a^{q-q^2} b^3 + b\right)\neq 0
\end{array}
\end{equation}
for all $b\in \gf_{2^n}^\ast$.
It seems that the sufficient and necessary condition here is more simple and compact, and
may be more easily handled.
\end{example}

In the end of this subsection, we revisit a class of pseudo-planar monomial proved by Scherr and Zieve.
For the readers' convenience, we recall their theorem.

\begin{Th}\label{th_ind3SZ}\cite{ScherrZieve}
For any positive integer $k$, write $q=2^{2k}$. If $c\in \gf_{q^3}^\ast$ is
a $(q-1)$-th power but not a $3(q-1)$-th power, then the function
$F(x)=cx^{q^2+q}$ is pseudo-planar over $\gf_{q^3}$.
\end{Th}

\begin{Prop}\label{prop_ind3Monmial}
Let $n=3m$, and let
$$F(x) = c x^{2^{2m}+2^m} \in \gf_{2^n}[x].$$
Assume that $c$ is a nonzero cube, and $c_0\in \gf_{2^n}^\ast$ such that $c_0^3=c$.
Set $q=2^m$ and $u=c_0^{-2(q^2+q+1)}$.  Then $F$ is pseudo-planar over $\gf_{2^n}$ if and only if
$u\neq 1$ and $$x^3+x^2+B_2x + \frac{B_2+1}{u+1}$$ is reducible over $\gf_q$ for any $B_2\in \gf_q$.
Particularly, if $m$ is even and $u=\omega$, where $\omega$ is an element with order 3,
then $F$ is a pseudo-planar function over $\gf_{2^n}$.
\end{Prop}

\begin{proof}
It follows from Proposition \ref{prop_ind3Trinomial} that $F$ is pseudo-planar if and only if
\begin{eqnarray*}
  \det M_a= a^{q^2+q+1} + \tr_{n/m}\left(c^2a^3\right)\neq 0,
\end{eqnarray*}
for all $a\in \gf_{2^n}^\ast$.
Let $a=c_0^{-2}b$. Then $a^3=c_0^{-6}b^3=c^{-2}b^3$, and $F$ is pseudo-planar if and only if
\begin{equation}\label{eq_ind3Mon}
    \det M_a =  ub^{q^2+q+1} + \tr_{n/m}(b^3)\neq 0
\end{equation}
for all $b\in \gf_{2^n}^\ast$, where $u=c_0^{-2(q^2+q+1)}\neq 0$.

If $b\in \gf_q^\ast$, then
$$\det M_a=(u+1)b^3,$$
which is nonzero if and only if $u\neq 1$.

In the following, we assume that $u\neq 1$ and $b\in \gf_{q^3}^\ast\setminus \gf_{q}$.
Let $B_1, B_2, B_3$ be defined as before.
Plugging $\tr_{n/m}(b^3)=B_1^3 + B_1B_2 + B_3$ into \eqref{eq_ind3Mon},
we have
$$\det M_a =  (u+1)B_3 + B_1^3 + B_1B_2. $$
We distinguish two cases.

{\bf Case 1: $B_1=0$.}

Then it is clear that $\det M_a=(u+1)B_3\neq 0$.

 {\bf Case 2: $B_1\neq 0$.}

WLOG, we assume that $B_1=1$. Then
$$\det M_a =  (u+1)B_3 + B_2 + 1. $$
Assume $\det M_a=0$ for some $b$. Then it
follows that
$$B_3 =\frac{B_2 + 1}{u+1}.$$
Let us consider the polynomial
\begin{equation}\label{eq_ind3_Mon_m0}
    m_b(x)=x^3 + x^2 + B_2x + \frac{B_2 + 1}{u+1}.
\end{equation}

According to the analysis in Section III.B, $F$ is pseudo-planar over $\gf_{2^n}$ if and only if
$u\neq 1$ and $m_b(x)$ is reducible over $\gf_{2^m}$ for any $B_2\in \gf_{2^m}$.

Now we prove the second part. Assume that $m$ is even and $u=\omega$, where $\omega$ is an element with order 3.
Then \eqref{eq_ind3_Mon_m0} turns to
\begin{eqnarray*}
   m_b(x) &=& x^3 + x^2 + B_2x + \omega (B_2 + 1) \\
   &=& (x+\omega) (x^2+\omega^2 x + B_2+1),
\end{eqnarray*}
which is reducible over $\gf_q$ for any $B_2\in \gf_q$. Hence $F$ is pseudo-planar over $\gf_{2^n}$.
\end{proof}

It can be easily verified that the condition in the last part of Proposition \ref{prop_ind3Monmial}, ie. $m$ is even and $u=\omega$,
 is equivalent to the sufficient condition in Theorem \ref{th_ind3SZ}. Hence we
 give another proof for Theorem \ref{th_ind3SZ}. Moreover, a sufficient and necessary
 condition for $F$ to be pseudo-planar is given here.

\subsection{Case 2: Extension Degree $t=4$}

\begin{Th}\label{th_ind4Gen}
Assume $n=4m$ and $q=2^m$. Let
\begin{eqnarray*}
  F(x) &=& \sum\limits_{i=0}^{3m-1} c_{1, i}x^{2^i(q+1)} + \sum\limits_{i=0}^{2m-1} c_{2, i}x^{2^i(q^2+1)} \\
  &&  +
\sum\limits_{i=0}^{m-1} c_{3, i}x^{2^i(q^3+1)}\in \gf_{2^n}[x].
\end{eqnarray*}
Then $F$ is pseudo-planar over $\gf_{2^n}$ if and only if
\begin{eqnarray*}
  & b^{q^3+q^2+q+1} + A_2^{2q+2} + (A_3^{2q^2+2} + A_3^{2q^3+2q}) \\
&\ \   + (b^{q^2+1}A_2^{2q} + b^{q^3+q}A_2^{2}) + \tr_{n/m}\left(b^{q^2+q}A_3^2\right)\neq 0
\end{eqnarray*}
for any nonzero $b$ in $\gf_{2^n}$, where
\begin{equation*}
    \left\{\begin{array}{lll}
             A_2 & = & \sum\limits_{i=0}^{2m-1}\left((c_{2,i}b)^{2^{n-i}} + (c_{2,i}b)^{2^{2m-i}}\right), \\
            A_3 & = & \sum\limits_{i=0}^{m-1}\left(c_{3,i}b\right)^{2^{n-i}} + \sum\limits_{i=0}^{3m-1}\left(c_{1,i}b\right)^{2^{3m-i}}.
           \end{array}
    \right.
\end{equation*}
\end{Th}

\begin{proof}
By Theorem \ref{th_Gen}, the dual linearized polynomial of $ \mathbb{L}_a(x) = F(x+a) + F(x) + F(a) + ax$ is
$\mathbb{L}^\ast_b(a)$:

\begin{eqnarray*}
  \mathbb{L}^\ast_b(a) &=& A_0\cdot a +  A_1\cdot a^{2^m} + A_2\cdot a^{2^{2m}} + A_3\cdot a^{2^{3m}},
\end{eqnarray*}
where
\begin{equation*}
    \left\{\begin{array}{lll}
             A_0 & = & b, \\
             A_1 &=&  \sum\limits_{i=0}^{3m-1}\left(c_{1,i}b\right)^{2^{n-i}} + \sum\limits_{i=0}^{m-1}\left(c_{3,i}b\right)^{2^{m-i}} = A_3^q,  \\
             A_2 & = & \sum\limits_{i=0}^{2m-1}\left((c_{2,i}b)^{2^{n-i}} + (c_{2,i}b)^{2^{2m-i}}  \right)\in \gf_{q^2}, \\
            A_3 & = & \sum\limits_{i=0}^{m-1}\left(c_{3,i}b\right)^{2^{n-i}} + \sum\limits_{i=0}^{3m-1}\left(c_{1,i}b\right)^{2^{3m-i}}.
           \end{array}
\right.
\end{equation*}
Hence
\begin{eqnarray*}
 \det M_b   &=& \left|\begin{array}{cccc}
                       A_0 & A_1 & A_2 & A_3 \\
                       A_3^q & A_0^q & A_1^q & A_2^q \\
                       A_2^{q^2} & A_3^{q^2} & A_0^{q^2} & A_1^{q^2} \\
                       A_1^{q^3} & A_2^{q^3} & A_3^{q^3} &  A_0^{q^3}
                     \end{array}
  \right|\\
  &=& \left|\begin{array}{cccc}
                       A_0 & A_3^q & A_2 & A_3 \\
                       A_3^q & A_0^q & A_3^{q^2} & A_2^q \\
                       A_2 & A_3^{q^2} & A_0^{q^2} & A_3^{q^3} \\
                       A_3 & A_2^q & A_3^{q^3} &  A_0^{q^3}
                     \end{array}
  \right|
\end{eqnarray*}
Then the result follows from Theorem \ref{th_Gen} and a direct computation.
\end{proof}

Similarly as in the extension degree 3 case,  we
set up some notations before constructing  pseudo-planar functions.
Let
$$x_1=b,\  x_2=b^q,\  x_3=b^{q^2}, \  \text{and } x_4=b^{q^3}.$$
Then \eqref{eq_B1Bt} and \eqref{eq_mb_Gen} become
 \begin{eqnarray*}
   B_1 &=& x_1 + x_2 + x_3 + x_4 = \tr_{n/m}(b), \\
    B_2 &=& x_1x_2 + x_1x_3 + x_1x_4 + x_2x_3 + x_2x_4 + x_3x_4,  \\
     B_3 &=& x_1x_2x_3 + x_1x_2x_4 + x_1x_3x_4 + x_2x_3x_4, \\
      B_4 &=& x_1  x_2  x_3  x_4 = {\rm{N}}_{n/m}(b),
\end{eqnarray*}
and
$$m_b(x)=x^4 + B_1x^3 + B_2x^2 + B_3x + B_4\in \gf_q[x]$$
respectively. 
%

\begin{Th}\label{th_ind4Class1}
Set $q=2^m$ and $n=4m$.  Let
$$F(x)= x^{q+1} + x^{q^2+1} + x^{q^3+q} + x^{q^3+1}.$$
Then $F$ is pseudo-planar over $\gf_{2^n}$.
\end{Th}
\begin{proof}
According to Theorem \ref{th_ind4Gen}, we have
\begin{equation*}
    \left\{\begin{array}{lll}
             A_2 & = & b + b^q + b^{q^2} + b^{q^3} = \tr_{n/m}(b), \\
            A_3 & = & b^{q^3} + b.
           \end{array}
    \right.
\end{equation*}
Then a direct computation shows

\begin{equation*}
    \left\{\begin{array}{l}
             A_2^{2q+2} = \tr_{n/m}(b^4), \\
            A_3^{2q^2+2} + A_3^{2q^3+2q} = \tr_{n/m}(b^{2q+2}),\\
            b^{q^2+1}A_2^{2q} + b^{q^3+q}A_2^{2} = (b^{q^2+1} + b^{q^3+q})\cdot \tr_{n/m}(b^2),\\
            \tr_{n/m}\left(b^{q^2+q}A_3^2\right) = \tr_{n/m}\left(b^{2q^3+q^2+q} + b^{q^2+q+2}\right).
           \end{array}
    \right.
\end{equation*}

Hence
\begin{eqnarray*}
  \det M_b &=& b^{q^3+q^2+q+1} + \tr_{n/m}(b^4) + \tr_{n/m}(b^{2q+2})\\
   &&  \ + \ (b^{q^2+1} + b^{q^3+q})\cdot \tr_{n/m}(b^2) \\
   && \ + \ \tr_{n/m}\left(b^{2q^3+q^2+q} + b^{q^2+q+2}\right).
\end{eqnarray*}

Then we have
\begin{equation}\label{eq_ind4Mb}
    \det M_b = B_4 + B_1^4 + B_2^2 + B_1B_3 + t_1,
\end{equation}
where
\begin{equation}\label{eq_ind4t1}
    t_1 = x_1^3x_3 + x_1^2x_3^2 + x_1x_3^3 + x_2^3x_4 + x_2^2x_4^2 + x_2x_4^3.
\end{equation}

Let $t_2$ and $t_3$ be the images of $t_1$ under the transformation of $(12)$ (or $(34)$) and $(14)$  (or $(23)$) respectively.
 \begin{eqnarray*}
   t_2 &=& x_1^3x_4 + x_1^2x_4^2 + x_1x_4^3 + x_2^3x_3 + x_2^2x_3^2 + x_2x_3^3. \\
    t_3 &=& x_1^3x_2 + x_1^2x_2^2 + x_1x_2^3 + x_3^3x_4 + x_3^2x_4^2 + x_3x_4^3.
\end{eqnarray*}

Then the following identities hold.
\begin{equation}\label{eq_ind4_tsum}
    t_1 + t_2 + t_3 =  B_1^2B_2 + B_1B_3+B_2^2,
\end{equation}
\begin{equation}\label{eq_ind4_tquad}
    t_1t_2 + t_1t_3 + t_2t_3 =   B_1^5B_3 + B_1B_2^2B_3+ B_1^2B_3^2,
\end{equation}
\begin{eqnarray}\label{eq_ind4_tmul}
    t_1 t_2  t_3 &=&  B_4B_1^8 + B_4B_1^6B_2 + B_4B_1^4B_2^2 \\
 & & \nonumber + B_4B_1^2B_2^3 + B_1^2B_2^2B_3^2 + B_1^3B_3^3 + B_2^3B_3^2.
\end{eqnarray}

Firstly, we assume that $b\in \gf_{q^2}^\ast$. Then we have $x_1=x_3$ and $x_2=x_4$.
Further, it follows that $B_1=B_3=0$, $B_2=x_1^2+x_2^2$ and $t_1=x_1^4+x_2^4=B_2^2$.
Hence
$$\det M_b = B_4\neq 0$$
for any $b\in \gf_{q^2}^\ast$. In the following, we always assume that $b\in \gf_{q^4}^\ast \setminus \gf_{q^2}$.
 Hence $m_b(x)$ is an irreducible polynomial over $\gf_q$ with degree 4.

The following proof is split into two cases according to $B_1=0$ or not.

{\bf Case 1:  $B_1=0$.}

Now \eqref{eq_ind4Mb} becomes
\begin{equation}\label{eq_ind4_B10_Mb}
\det M_b = B_4 + B_2^2 + t_1,
\end{equation}
and \eqref{eq_ind4_tsum}, \eqref{eq_ind4_tquad} and \eqref{eq_ind4_tmul}
reduce to

\begin{equation}\label{eq_ind4_B10_tsum}
    t_1 + t_2 + t_3 =  B_2^2,
\end{equation}
\begin{equation}\label{eq_ind4_B10_tquad}
    t_1t_2 + t_1t_3 + t_2t_3 =  0,
\end{equation}
\begin{equation}\label{eq_ind4_B10_tmul}
    t_1 t_2  t_3 = B_2^3B_3^2.
\end{equation}

Assume, on the contrary,  that $\det M_b = 0$ for some $b\in \gf_{q^4}^\ast \setminus \gf_{q^2}$.
Then it follows from \eqref{eq_ind4_B10_Mb} that
\begin{equation}\label{eq_B10_t1}
    t_1 = B_4 + B_2^2.
\end{equation}
Plugging it into \eqref{eq_ind4_B10_tsum}, one gets
\begin{equation}\label{eq_B10_t23S}
    t_2 + t_3 = B_4.
\end{equation}
With \eqref{eq_ind4_B10_tquad}, we deduce that
\begin{equation}\label{eq_B10_t23M}
    t_2t_3 = t_1 (t_2 + t_3)=  B_4^2 + B_2^2B_4.
\end{equation}
Substituting \eqref{eq_B10_t1} and \eqref{eq_B10_t23M} into
\eqref{eq_ind4_B10_tmul} leads to
\begin{equation}\label{eq_ind4_B10}
    B_4^3 + B_2^4B_4 + B_2^3B_3^2=0.
\end{equation}
Since $B_4\neq 0$, we know $B_2\neq 0$.

Define $r=(B_4/B_2)^{1/2}$, $u=B_2$ and $v=B_4/B_2=r^2$. Then $u, v, r\in \gf_q$.
Now we compute
\begin{eqnarray*}
   && (x^2+rx+u)(x^2+rx+v) \\
   &=&  x^4 + (r^2+u+v)x^2 + r(u+v)x + uv\\
   &=&  x^4 + B_2x^2 + \left((B_4 + B_2^2)B_4^{1/2}/B_2^{3/2}\right)x + B_4\\
   &=&  x^4 + B_2x^2 + \left((B_4^3 + B_2^4B_4)/B_2^{3}\right)^{1/2}x + B_4\\
   &=&  x^4 + B_2x^2 + B_3x + B_4\\
   &=& m_b(x),
\end{eqnarray*}
where the last second equality follows from \eqref{eq_ind4_B10}. Thus
$m_b(x)$ can be factored into two quadratic polynomials over $\gf_{q}$, which is impossible.
Hence $\det M_b\neq 0$ if $B_1=0$.

{\bf Case 2:  $B_1\neq 0$.}

WLOG, we assume that $B_1=1$. Then  \eqref{eq_ind4Mb} becomes
\begin{equation}\label{eq_ind4_B11_Mb}
\det M_b = B_4  + B_2^2 + B_3 + t_1 + 1,
\end{equation}
and \eqref{eq_ind4_tsum}, \eqref{eq_ind4_tquad} and \eqref{eq_ind4_tmul}
reduce to
\begin{equation}\label{eq_ind4_B11_tsum}
    t_1 + t_2 + t_3 =  B_2 + B_3+B_2^2,
\end{equation}
\begin{equation}\label{eq_ind4_B11_tquad}
     t_1t_2 + t_1t_3 + t_2t_3 =   B_3 + B_2^2B_3+ B_3^2,
\end{equation}
\begin{eqnarray}\label{eq_ind4_B11_tmul}
    t_1 t_2  t_3 &=&  B_4 + B_4B_2 + B_4B_2^2 + B_4B_2^3\\
 & & \nonumber  + B_2^2B_3^2 + B_3^3 + B_2^3B_3^2.
\end{eqnarray}

Assume, on the contrary,  that $\det M_b = 0$ for some $b\in \gf_{q^4}^\ast \setminus \gf_{q^2}$.
Then it follows from \eqref{eq_ind4_B11_Mb} that
\begin{equation}\label{eq_B11_t1}
    t_1 = B_4  + B_2^2 + B_3 + 1.
\end{equation}

Plugging it into \eqref{eq_ind4_B11_tsum}, one gets
\begin{equation}\label{eq_B11_t23S}
    t_2 + t_3 = B_2 + B_3+B_2^2 +t_1 = B_4 + B_2 +1.
\end{equation}
With \eqref{eq_ind4_B11_tquad}, we deduce that
\begin{eqnarray}\label{eq_B11_t23M}
   t_2t_3& =& B_3 + B_2^2B_3+ B_3^2 + t_1 (t_2 + t_3)\\
\nonumber     &=&  B_4^2 + B_3B_4 + B_2^2B_4 + B_2B_4 + B_3^2\\
 \nonumber    & &\   + B_2^2B_3 + B_2B_3 + B_2^3 + B_2^2 +B_2 +1.
\end{eqnarray}
Substituting \eqref{eq_B11_t1} and \eqref{eq_B11_t23M} into
\eqref{eq_ind4_B11_tmul}, and after a direct computation, we finally get
\begin{equation}\label{eq_ind4_B11}
    B_4^3 + (B_2+1)B_4^2 + C_1B_4 + C_0=0,
\end{equation}
where
\begin{eqnarray*}
  C_1 &=& B_3(B_2+1)^2 +  B_2(B_2+1)^3, \\
  C_0 &=& B_3^2(B_2+1)^3 +  B_3(B_2+1)^4 + (B_2+1)^5.
\end{eqnarray*}
Combing the above equation with $B_4\neq 0$, one can conclude that $B_2+1\neq 0$.

Let
\begin{equation}\label{eq_ind4B42}
    B_4=(B_2+1)(z+1).
\end{equation}
 Plugging it into \eqref{eq_ind4_B11}, then dividing $(B_2+1)^3$ across the both sides,
 and  after simplification, we have
 \begin{equation}\label{eq_ind4_z}
    z^3 + (B_2^2+B_3+B_2+1)z + (B_3^2 + B_2B_3 +B_2+1)=0.
\end{equation}

In the rest of the proof,  we distinguish two subcases.

 {\bf Subcase 2.1: $B_4 = B_2^2+1$. }

 Plugging $B_4 = B_2^2+1$ into \eqref{eq_ind4B42},
 one can deduce that $z=B_2$, and then substituting it into \eqref{eq_ind4_z} leads to $B_3=B_2+1$.
  Let $r$ be an element of $\gf_{q^2}$ such that
 $r^2+r+B_2=0$. Define
 $$\phi(x) = x^2 + r x + (B_2+1).$$
 Since
 \begin{eqnarray*}
   && \  \tr_{2m/1}\left(\frac{B_2+1}{r^2}\right) = \tr_{2m/1}\left(\frac{r^2+r+1}{r^2}\right)\\
   &  = & \tr_{2m/1}\left(1+\frac{1}{r}+\frac{1}{r^2}\right)=0,
 \end{eqnarray*}
$\phi(x)$ is reducible over $\gf_{q^2}$ according to Lemma \ref{redu-condition}.
 Let $\tau\in \gf_{q^2}$ be a zero of $\phi(x)$. Then
 $$\tau^2 + r\tau + (B_2+1)=0.$$
 Now we compute
\begin{eqnarray*}
    & & m_b(\tau) \\
    &=&\tau^4 + \tau^3 + B_2\tau^2 + B_3\tau + B_4    \\
   &=&  \tau^4 + \tau^3 + (r^2+r)\tau^2 + (B_2+1)\tau + (B_2+1)^2\\
   &= & \left(\tau^2 + r\tau + (B_2+1)\right)^2  + \tau \left(\tau^2 + r\tau + (B_2+1)\right)  \\
   &=&  0.
\end{eqnarray*}
Thus $\tau\in \gf_{q^2}$ is a zero of $m_b(x)$, which contradicts the assumption that $m_b(x)$ is
irreducible over $\gf_q$.

 {\bf Subcase 2.2: $B_4\neq B_2^2+1$. }

Let us define
\begin{equation}\label{eq_ind4uv}
    u=B_2+1,  v=B_4/(B_2+1).
\end{equation}
Then $u\neq v$.  Set
 \begin{equation}\label{eq_ind4_r}
    r=\frac{B_3+u}{u+v}=\frac{B_2^2+B_2B_3+B_3+1}{B_2^2+B_4+1}.
 \end{equation}
    Then $u, v, r\in \gf_q$
 and
 \begin{equation}\label{eq_ind4_r1}
   r+1=\frac{B_3+v}{u+v}=\frac{B_2B_3+B_3+B_4}{B_2^2+B_4+1}.
 \end{equation}
Hence
\begin{eqnarray*}
   && (x^2+rx+u)(x^2+(r+1)x+v) \\
   &=&  x^4 + x^3 + (r(r+1)+u+v)x^2 \\
   && \ \ + ( (r+1)u + rv )x + uv\\
   &=&  x^4 + x^3 + (r(r+1)+u+v)x^2  \\
    && \ \ + \left( \frac{(B_3+v) u + (B_3+u) v }{u+v} \right)x + B_4 \\
   &=&  x^4 + x^3 + (r(r+1)+u+v)x^2 + B_3x + B_4.
\end{eqnarray*}
Now, to finish the proof, it suffices to prove that
\begin{equation}\label{eq_ind4_ruv}
    r(r+1)+u+v=B_2,
\end{equation}
which means that $m_b(x)$ can be factored into two polynomials with degree $2$ over $\gf_q$,
and it will then lead to a contradiction.

Plugging \eqref{eq_ind4uv}, \eqref{eq_ind4_r} and \eqref{eq_ind4_r1} into \eqref{eq_ind4_ruv} leads to
\begin{equation*}
    \frac{(B_2^2+B_2B_3+B_3+1)(B_2B_3+B_3+B_4)}{(B_2^2+B_4+1)^2} = \frac{B_4}{B_2+1} + 1.
\end{equation*}
Substituting \eqref{eq_ind4B42} into the above equation, we have
\begin{equation*}
    \frac{(B_2+B_3+1)(B_3+z+1)}{(B_2+z)^2} = z,
\end{equation*}
which can be easily verified to be equivalent to \eqref{eq_ind4_z}. Hence \eqref{eq_ind4_ruv}
always holds.

We finish the proof.
\end{proof}

\begin{Th}\label{th_ind4Class2}
Set $q=2^m$ and $n=4m$.  Let
$$F(x)=x^{q^2+q} + x^{q^3+q^2} + x^{q^3+q}.$$
Then $F$ is pseudo-planar over $\gf_{2^n}$.
\end{Th}
\begin{proof}
By Theorem \ref{th_ind4Gen}, we have
\begin{equation*}
    \left\{\begin{array}{lll}
             A_2 & = &  b^{q^3}+b^q, \\
            A_3 & = & b^{q^2} + b^q.
           \end{array}
    \right.
\end{equation*}
Then a lengthy but direct computation shows that
\begin{eqnarray*}
   && \det M_b\\
  &=& b^{q^3+q^2+q+1} + \tr_{n/m}\left(b^{q^2+3} + b^{q^2+3q} + b^{3q^2+q}\right) \\
   &=&  B_4 + B_1^2B_2 + B_1B_3.
\end{eqnarray*}
If $B_1=0$, then $\det M_b=B_4\neq 0$.
If $b\in \gf_{q^2}^\ast$, then $B_1=0$. In the following,
we assume that $b\in \gf_{q^4}^\ast \setminus \gf_{q^2}$ and $B_1\neq 0$.
WLOG, let $B_1=1$. Assume that
$$\det M_b=B_4 + B_2 + B_3=0$$
for some $b\in \gf_{q^4}^\ast \setminus \gf_{q^2}$. Then
$B_4=B_2+ B_3$, and
\begin{eqnarray*}
  m_b(x) &=& x^4 + B_1x^3 + B_2x^2 + B_3x + B_4 \\
   &=& x^4 + x^3 + B_2x^2 + B_3x + B_2 + B_3 \\
   &=& (x+1) (x^3 + B_2x + B_2 + B_3).
\end{eqnarray*}
Contradicts! We finish the proof.
\end{proof}

\subsection{Case 3: Extension Degree $t=2$}

\begin{Th}\label{th_ind2Gen}
Let $n=2m$, and let
$$F(x)=\sum\limits_{i=0}^{m-1} c_ix^{2^{m+i}+2^i}\in \gf_{2^n}[x].$$
Then $F$ is pseudo-planar over $\gf_{2^n}$ if and only if
$$b^{2^m+1} + \sum\limits_{i=0}^{m-1}(c_ib)^{2^{m-i+1}} + \sum\limits_{i=0}^{m-1}(c_ib)^{2^{2m-i+1}}\neq 0$$
for any nonzero $b$ in $\gf_{2^n}$.
\end{Th}

\begin{proof}
Set $q=2^m$. According to Theorem \ref{th_Gen}, the dual linearized polynomial of $ \mathbb{L}_a(x) = F(x+a) + F(x) + F(a) + ax$ is
$\mathbb{L}^\ast_b(a)$:
\begin{eqnarray*}
  \mathbb{L}^\ast_b(a) &=& A_0\cdot a +  A_1\cdot a^{2^m},
\end{eqnarray*}
where
\begin{equation*}
    \left\{\begin{array}{lll}
             A_0 & = & b, \\
             A_1 &=&  \sum\limits_{i=0}^{m-1}\left(c_ib\right)^{2^{n-i}} + \sum\limits_{i=0}^{m-1}\left(c_{i}b\right)^{2^{m-i}}\in \gf_q.
           \end{array}
\right.
\end{equation*}
Hence
$$\det M_b  = \left|\begin{array}{cc}
                       A_0 & A_1  \\
                       A_1^q & A_0^q
                     \end{array}
  \right|= A_0^{q+1} + A_1^{q+1} = b^{q+1} + A_1^2.$$
Then the result follows from Theorem \ref{th_Gen}.
\end{proof}

Now we use Theorem \ref{th_ind2Gen} to characterize a monomial
pseudo-planar function, which was firstly studied by Schmidt and Zhou in \cite{SchmidtZhou}.

\begin{Th}\label{th_ind2Monomial}
Let $n=2m$, and let
$$F(x)=cx^{2^m+1}, \ \ \text{where} \ c\in \gf_{2^n}.
$$
Then $F$ is pseudo-planar over $\gf_{2^n}$ if and only if $\tr_{m/1}(c^{2^m+1})=0$.
Further, the number of such $c$ in $\gf_{2^n}$ is equal to  $2^{2m-1}-2^{m-1}$.
\end{Th}
\begin{proof}
We only prove the sufficient and necessary condition in the first part.
Then the counting argument follows directly.

Let $q=2^m$. The case $c=0$ is trivial. We assume in the following that $c\neq 0$.
According to Theorem \ref{th_ind2Gen}, $F$ is pseudo-planar if and only if
\begin{equation}\label{eq_ind2_mon}
    \det M_a = a^{q+1} + (ca)^2 + (ca)^{2q}\neq 0
\end{equation}
for any nonzero $a\in \gf_{2^n}^\ast$.
Let $a=c^{-1}b$. Define $x_1=b$ and $x_2=b^q$. Let
\begin{eqnarray*}
  B_1 &=& x_1+x_2=b+b^q=\tr_{n/m}(b),  \\
  B_2 &=& x_1x_2=b^{q+1}={\rm N}_{n/m}(b).
\end{eqnarray*}
Then $F$ is pseudo-planar if and only if
\begin{eqnarray*}
 \det M_a   &=& c^{-(q+1)} b^{q+1} + b^2 + b^{2q} \\
 &=& c^{-(q+1)} B_2 + B_1^2\neq 0
\end{eqnarray*}
for any nonzero $b\in \gf_{2^n}^\ast$.

If $b\in \gf_q^\ast$, then
$$\det M_a=c^{-(q+1)} b^2,$$
which is clearly nonzero for any nonzero $b$.

In the following, we assume that $b\in \gf_{q^2}^\ast\setminus \gf_{q}$.
We distinguish two cases.

{\bf Case 1: $B_1=0$.}

Then it is clear that $\det M_a=c^{-(q+1)} B_2\neq 0$.

 {\bf Case 2: $B_1\neq 0$.}

WLOG, we assume that $B_1=1$. Then
$$\det M_a =  c^{-(q+1)} B_2 + 1. $$
Assume $\det M_a=0$ for some $b$. Then it
follows that
$$B_2 =  c^{q+1}.$$
Let us consider the polynomial
\begin{equation}\label{eq_ind3_Mon_m}
    m_b(x) = x^2 + x + c^{q+1}.
\end{equation}
If $\tr_{m/1}(c^{q+1})\neq 0$, then $m_b(x)$ is irreducible over $\gf_q$. Hence
its solutions are all in $ \gf_{q^2}^\ast\setminus \gf_{q}$,
and for each solution, $\det M_a=0$ holds, which means that
$F$ is not pseudo-planar. On the other hand, if $\tr_{m/1}(c^{q+1})=0$, then $m_b(x)$ is reducible over $\gf_q$,
which contradicts that $b\in \gf_{2^n}^\ast\setminus \gf_{q}$. This contradiction shows that $\det M_a\neq 0$ holds.
Hence $F$ is  pseudo-planar over $\gf_{2^n}$.
\end{proof}

The above theorem generalizes \cite[Theorem 3.1]{SchmidtZhou}, which said that: if $c\in \gf_q^\ast$ and $\tr_{m/1}(c)=0$, then
$F(x)=cx^{q+1}$ is pseudo-planar over $\gf_{q^2}$.

An exhaustive search over $\gf_{2^{2m}}$ for $1\leq m\leq 4$ shows that there are no
pseudo-planar functions with the form $\sum\limits_{i=0}^{m-1} c_ix^{2^{m+i}+2^i}$, where
$c_i\in \gf_{2^{2m}}$ other than the monomials given by Theorem \ref{th_ind2Monomial}.
It takes about $120$ hours for the exhaustive search over $\gf_{2^8}$ by Magma V2.12-16 on a personal computer
(IntelCore CPU i5-3337U@1.80GHz, 1.80GHz, RAM 8.0GB).
 Hence
we propose the following conjecture. We can not prove it now and leave it as an open problem.

\begin{Prob}\label{conj_ind2}
Set $n=2m$ and $q=2^m$.  Let
$$F(x)=\sum\limits_{i=0}^{m-1} c_ix^{2^{m+i}+2^i}\in \gf_{2^n}[x].$$
To prove $F$ is pseudo-planar over $\gf_{2^n}$ if and only if
$\tr_{m/1}(c_0^{q+1})=0$, and $c_1=c_2=\cdots=c_{m-1}=0$;
or to find a counter-example.
\end{Prob}

\section{Equivalence Problem on  constructed pseudo-planar functions}

In Section III a general family of quadratic pesudo-planar functions was presented.
Moreover, in Section IV five explicit families of pesudo-planar functions were constructed.
Note that we call a family of pesudo-planar functions explicit if the condition (for it to be pesudo-planar) can be easily verified.
For example, the following are the list of these five explicit families of functions,
while the family defined by Proposition \ref{prop_ind3Trinomial} is not  explicit since
the condition \eqref{eq_ind3TriCond} can not be easily verified (though it can be verified by computer for small variables).
\begin{enumerate}
  \item $cx^{2(q+1)} + c^qx^{2(q^2+1)}$, where  $n=3m$, $q=2^m$, $c\in \gf_{2^n}$ (Theorem \ref{th_ind3Class1}).
  \item $x^{2(q+1)} + x^{q^2+1} + x^{q^2+q} +  x^{2(q^2+1)}$, where  $n=3m$, $m\not\equiv 1\mod 3$ and  $q=2^m$ (Theorem \ref{th_ind3Class2}).
  \item $x^{q+1} + \alpha x^{q^2+q} +  x^{q^2+1}$, where   $n=3m$, $q=2^m$ and $\alpha^3+\alpha^2+1=0$ (Corollary \ref{cor_ind3Class3}).
  \item $x^{q+1} + x^{q^2+1} + x^{q^3+q} + x^{q^3+1}$, where  $n=4m$, $q=2^m$ (Theorem \ref{th_ind4Class1}).
  \item $x^{q^2+q} + x^{q^3+q^2} + x^{q^3+q}$, where  $n=4m$, $q=2^m$ (Theorem \ref{th_ind4Class2}).
\end{enumerate}

In this section, we will discuss the equivalence problem on these functions.
Firstly, the pesudo-planar functions in Theorem
\ref{th_ind3Class1}, Theorem \ref{th_ind3Class2} and Corollary \ref{cor_ind3Class3}  cannot be new.
The reason is that they are all of Dembowski-Ostrom type, which means that the semifields' centers must contain $\gf_q$.
By the classification of semifields of order $q^3$ over $\gf_q$ by Menichetti in 1977 \cite{Menichetti}, they must be finite fields.
Therefore these functions should be equivalent to $F(x)=0$. The same argument also works for the functions in Result \ref{result_binom}
discovered by Hu et al \cite{HuLiZhang}.

Secondly, we study the equivalence of the functions in Theorems \ref{th_ind4Class1} and \ref{th_ind4Class2}.
To check whether they are new or not, we determine the left (right) nucleus of the derived semifields.

\begin{Prop}
Let $F$ be the function in Theorem \ref{th_ind4Class1} or Theorem \ref{th_ind4Class2}.
Then the semifield derived from $F$ is isomorphic to the finite field.
\end{Prop}

\begin{proof}
We only prove the case that $F$ is the function in Theorem \ref{th_ind4Class1}.
The other case can be proved similarly and is omitted here.
Then
$$F(x)= x^{q+1} + x^{q^2+1} + x^{q^3+q} + x^{q^3+1},$$
where $q=2^m$ and $n=4m$.

Let us define the following multiplication
\begin{eqnarray*}
   x\ast y
  &=& xy + F(x+y) + F(x) + F(y) \\
   &=&  x\tr_{n/m}(y) + x^q(y+y^{q^3}) + x^{q^2}y + x^{q^3} (y+y^q).
\end{eqnarray*}
Since $x\ast 1=x^{q^2}$, $(\gf_{2^n}, +, \ast)$ is not a semifield but a presemifield.
Then we define
\begin{eqnarray*}
  & & x\circ y = (x\ast y)^{q^2} \\
  & =& xy^{q^2} + x^{q} (y^{q^2}+y^{q^3}) + x^{q^2}\tr_{n/m}(y) + x^{q^3}(y^{q^2}+y^{q}).
\end{eqnarray*}
Hence $(\gf_{2^n}, +, \circ)$ is  a semifield corresponding to $F$.

On one hand, we have
\begin{eqnarray*}
   && a\circ (x\circ y) \\
   &=& a A_0(x, y) +  a^q A_1(x, y) + a^{q^2} A_2(x, y) + a^{q^3} A_3(x, y),
\end{eqnarray*}
where
\begin{eqnarray*}
   A_0(x, y)&=& (x\circ y)^{q^2},  \\
  A_1(x, y)&=& \left((x\circ y)^{q^2}+(x\circ y)^{q^3}\right),  \\
  A_2(x, y)&=& \tr_{n/m}(x\circ y),  \\
  A_3(x, y)&=& \left((x\circ y)^{q^2}+(x\circ y)^{q}\right).
\end{eqnarray*}

On the other hand, we have
\begin{eqnarray*}
   &&(a\circ x) \circ y\\
   &=& (a\circ x)y^{q^2} + (a\circ x)^{q} (y^{q^2}+y^{q^3}) \\
   && \ \ + (a\circ x)^{q^2}\tr_{n/m}(y) + (a\circ x)^{q^3}(y^{q^2}+y^{q})\\
   &=& a B_0(x, y) +  a^q B_1(x, y) + a^{q^2} B_2(x, y)\\
    && \ \ + a^{q^3} B_3(x, y),
\end{eqnarray*}
where
\begin{eqnarray*}
   && B_0(x, y)\\
   &=&  x^{q^2}y^{q^2} + (x^{q^2} + x^q)^{q} (y^{q^2}+y^{q^3})\\
   && \ \  + \tr_{n/m}(x)\tr_{n/m}(y) + (x^{q^2} + x^{q^3})^{q^3}(y^{q^2}+y^{q}), \\
   && B_1(x, y)\\
   &=&  (x^{q^2} + x^{q^3}) y^{q^2} + (x^{q^2})^q (y^{q^2}+y^{q^3}) \\
   && \ \ +  (x^{q^2} + x^q)^{q^2}\tr_{n/m}(y) + \tr_{n/m}(x)(y^{q^2}+y^{q}), \\
   && B_2(x, y)\\
   &=&  \tr_{n/m}(x)y^{q^2} + (x^{q^2} + x^{q^3})^{q} (y^{q^2}+y^{q^3}) \\
   && \ \ + (x^{q^2})^{q^2}\tr_{n/m}(y) + (x^{q^2} + x^q)^{q^3}(y^{q^2}+y^{q}), \\
    && B_3(x, y)\\
   &=&  (x^{q^2} + x^q)y^{q^2} + \tr_{n/m}(x) (y^{q^2}+y^{q^3}) \\
   && \ \ + (x^{q^2} + x^{q^3})^{q^2}\tr_{n/m}(y) + (x^{q^2})^{q^3}(y^{q^2}+y^{q}).
\end{eqnarray*}
Then a direct computation shows that
$$A_i(x, y) = B_i(x, y), i=0, 1, 2, 3.$$
Hence
$$a\circ (x\circ y)=(a\circ x) \circ y \ \text{ for all} \ a, x,y \in \gf_{2^n}, $$
which means that  $(\gf_{2^n}, +, \circ)$ is isomorphic to the finite field $\gf_{2^n}$.
\end{proof}

It is a pity that all the explicit families of pesudo-planar functions constructed in the last section are
equivalent to $F(x)\equiv 0$. However,  they are still interesting since it may be hard to prove a given function to be pesudo-planar even
 if it is equivalent to known functions.
 For example, the pesudo-planar function in \cite[Theorem 1.1]{ScherrZieve} is equivalent to the zero function.
 However, the fact that it is pesudo-planar seems not to be easily proved. The functions in Result \ref{result_binom} are also such examples.

Since the number of pairwise nonisomorphic commutative semifields of even order $N$ in the Kantor family
is not bounded above by any polynomial in $N$,
and the Kantor family is  included in the general family constructed in Theorem \ref{th_Gen}
(as shown in  the end of Section III.A),
we know that there exist plenties of pesudo-planar functions in our general family
which are inequivalent to the zero function.
However, we are wondering whether there exists a function in Theorem \ref{th_Gen}
 which is inequivalent to all known pesudo-planar functions. Currently we can not find
an answer and leave it as an open problem.

\begin{Prob}\label{conj_new}
Does there exist a pesudo-planar function in
the general family given by Theorem \ref{th_Gen} which is inequivalent to
those in Result \ref{result1}?  If yes, find such an example.
\end{Prob}

\section{Applications of constructed pseudo-planar functions}


%

According to Theorem \ref{th_Abdukhalikov} and Proposition \ref{prop_codebook}, 
the pseudo-planar functions constructed in Section IV
can contribute a lot of complete sets of MUBs, optimal codebooks meeting the Levenstein
bound. They can also be used to construct compressed sensing matrices with low coherence.
In the following we give a small example over $\gf_{2^3}$.

\begin{example}\label{ex_new}
In Theorem \ref{th_ind3Class1}, set $m=1$, $n=3$ and $c=1$. Then $F(x)=x^6+x^{10}$ is a pseudo-planar function over $\gf_{2^3}$.
According to Theorem \ref{th_Abdukhalikov} and Proposition \ref{prop_codebook}, the following bases is a complete set of MUB with dimension $3$.
 The union set of these basis vectors is an optimal $(72, 8)$  complex codebook meeting the Levenstein bound.
\begin{equation*}
    \begin{array}{rl}
B_1=\{(AAAAAAAA), & (AACACCCA),\\
(ACACCCAA), & (AACCCAAC),\\
(ACCCAACA), & (ACCAACAC),\\
(ACAACACC), & (AAACACCC)\},\\
B_2=\{(ADBDAADC), & (ADDDCCBC),\\
(ABBBCCDC), & (ADDBCADA),\\
(ABDBAABC), & (ABDDACDA),\\
(ABBDCABA), & (ADBBACBA),\\
\end{array}
\end{equation*}
\begin{equation*}
    \begin{array}{rl}
    B_3=\{(AADABDDC), & (AABADBBC)\},\\
(ACDCDBDC), & (AABCDDDA),\\
(ACBCBDBC), & (ACBABBDA),\\
(ACDADDBA), & (AADCBBBA)\},\\
\end{array}
\end{equation*}
\begin{equation*}
    \begin{array}{rl}
    B_4=\{(ADDCADAB), & (ADBCCBCB),\\
(ABDACBAB), & (ADBACDAD),\\
(ABBAADCB), & (ABBCABAD),\\
(ABDCCDCD), & (ADDAABCD)\},\\
\end{array}
\end{equation*}
\begin{equation*}
    \begin{array}{rl}
    B_5=\{(ABADDDAC), & (ABCDBBCC),\\
(ADABBBAC), & (ABCBBDAA),\\
(ADCBDDCC), & (ADCDDBAA),\\
(ADADBDCA), & (ABABDBCA)\},\\
\end{array}
\end{equation*}
\begin{equation*}
    \begin{array}{rl}
    B_6=\{(ADAADCDB), & (ADCABABB),\\
(ABACBADB), & (ADCCBCDD),\\
(ABCCDCBB), & (ABCADADD),\\
(ABAABCBD), & (ADACDABD)\},\\
\end{array}
\end{equation*}
\begin{equation*}
    \begin{array}{rl}
    B_7=\{(AADDDACB), & (AABDBCAB),\\
(ACDBBCCB), & (AABBBACD),\\
(ACBBDAAB), & (ACBDDCCD),\\
(ACDDBAAD), & (AADBDCAD)\},\\
\end{array}
\end{equation*}
\begin{equation*}
    \begin{array}{rl}
    B_8=\{(ACCBCBBD), & (ACABADDD),\\
(AACDADBD), & (ACADABBB),\\
(AAADCBDD), & (AAABCDBB),\\
(AACBABDB), & (ACCDCDDB)\},\\
\end{array}
\end{equation*}
\begin{equation*}
    \begin{array}{rl}
    B_{\infty}=\{(10000000), & (01000000),\\
(00100000), & (00010000),\\
(00001000), & (00000100),\\
(00000010), & (00000001)\},\\
\end{array}
\end{equation*}
where $A$, $B$, $C$ and $D$ denotes $\frac{1}{\sqrt{8}}$, $\frac{\sqrt{-1}}{\sqrt{8}}$,
$-\frac{1}{\sqrt{8}}$ and  $-\frac{\sqrt{-1}}{\sqrt{8}}$ respectively.
\end{example}

\section{Conclusion}

In this paper, we introduced a new approach to  constructing  quadratic pseudo-planar functions over $\gf_{2^n}$.
By using it,  a general family of such functions was constructed.
Then  five explicit  families of pseudo-planar functions were presented, and
many known families were reconstructed, some of which were generalized.
These pseudo-planar functions not only lead to projective planes,
relative difference sets and presemifields, but also give optimal codebooks meeting the Levenstein bound,
complete sets of MUB, and compressed sensing matrices with low coherence.

Now all the families of known pesudo-planar functions  are subfamilies of the functions with
the general form \eqref{eq_gen_Form}. On one hand, we believe that
there exist other explicit subfamilies of pseudo-planar functions in this general family.
Particularly, we are wondering  whether the answer to Problem \ref{conj_new} is positive.
On the other hand, it is more interesting to find a class of pseudo-planar functions
out of this family.
Further, we would like to ask again the following problem which was raised in \cite{PottSchmidtZhou}.

\begin{Prob}\label{conj_non_quad}
Is it possible to find a pesudo-planar function that is not of
Dembowski-Ostrom type?
\end{Prob}

To prove a quadratic function to be pseudo-planar, it is equivalent to proving
 a series of linearized polynomials are permutation polynomials. In this paper,
 instead of investigating these linearized polynomials directly, we turned to study the dual polynomials
  of these functions. It seems that this method is efficient. It should be useful to study other problems about linearized permutation polynomials.
Particularly, it may work for planar functions over finite fields with odd characteristic.


\section*{Acknowledgement}
The author would like to thank Dr. Yue Zhou for reading an early version of this paper
and giving many helpful
comments which improve the quality of the paper.

\bigskip

{\bfseries Longjiang Qu} received his B.A. degree in 2002 and Ph.D. degree in 2007 in
mathematics from the National University of Defense Technology, Changsha,
China. He is now a Professor with College of Science,
 National University of Defense Technology of China. His
research interests are cryptography and coding theory.


\begin{thebibliography}{99}

\bibitem{Abdukhalikov}
K. Abdukhalikov, ``Symplectic spreads, planar functions, and mutually unbiased bases,''
 J. Algebraic Comb., vol. 41, pp. 1055-1077, 2015.

\bibitem{CandesTao}
E. J. Cand\`{e}s and T. Tao, ``Decoding by linear programming,'' IEEE
Trans. Inf. Theory, vol. 51, no. 12, pp. 4203-4215, Dec. 2005.


\bibitem{CDY}
C. Carlet, C. Ding, and J. Yuan,  ``Linear codes from perfect nonlinear mappings and their secret sharing
schemes,'' IEEE Trans. Inf. Theory,  vol. 51, pp. 2089-2102, 2005.

\bibitem{Dembowski}
P. Dembowskii, Finite geometries, Springer, Berlin, 1968.

\bibitem{Ding06}
C. Ding, ``Complex codebooks from combinatorial designs,'' IEEE
Trans. Inf. Theory, vol. 52, no. 9, pp. 4229-4235, 2006.

\bibitem{DingFeng07}
C. Ding and T. Feng, ``A generic construction of complex codebooks
meeting the Welch bound,'' IEEE Trans. Inf. Theory, vol. 53, no. 11, pp.
4245-4250, 2007.

\bibitem{DingFeng08}
C. Ding and T. Feng, ``Codebooks from almost difference sets,'' Designs,
Codes Cryptogr., vol. 46, pp. 113-126, 2008.

\bibitem{DXY}
C. Ding, Q. Xiang, J. Yuan, et. al., ``Explicit classes of permutation polynomials of
$\gf_{3^{3m}}$,'' Science in China Series A: Mathematics, vol. 53, no. 4, pp. 639-647, 2009.

\bibitem{DingYin}
C. Ding and J. Yin, ``Signal sets from functions with optimum nonlinearity,'' IEEE Trans. Commun. vol. 55, no. 5,
pp. 936-940, 2007.

\bibitem{Donoho}
D. L. Donoho, ``Compressed sensing,'' IEEE Trans. Inf. Theory, vol. 52,
no. 4, pp. 1289-1306, 2006.


\bibitem{GanleySpence}
M.J. Ganley and E. Spence, ``Relative difference sets and quasiregular collineation groups,'' J. Combin.
Theory Ser. A,  vol. 19, pp. 134-153, 1975.

\bibitem{HKCSS94}
J. Hammons, P.V. Kumar, A.R. Calderbank, N.J. Sloane, P. Sol\'{e},  ``The $\mathbb{Z}_4$-linearity of
Kerdock, Preparata, Goethals, and related codes,'' IEEE Trans. Inf. Theory, vol. 40, no. 2, pp. 301-319, 1994.


\bibitem{HuLiZhang}
S. Hu, S. Li, T. Zhang, et. al., ``New pseudo-planar binomials in characteristic two
and related schemes,''  Des. Codes Cryptogr., vol. 76, pp. 345-360, 2015.


\bibitem{KabatyanskiiLevenshtein}
G. A. Kabatyanskii and V. I. Levenshtein, ``Bounds for packing on a
sphere and in space,'' Probl. Inf. Transmission, vol. 14, pp. 1-17, 1978.

\bibitem{Kantor}
W. M. Kantor, ``Commutative semifields and symplectic spreads,'' J.  Algebra,
vol. 270, no. 1, pp.96-114, 2003.

\bibitem{LavrauwPolverino}
M. Lavrauw and O. Polverino, ``Finite semifields and Galois geometry,'' In: De Beule J., Storme L. (eds.)
Current Research Topics in Galois Geometry, NOVA Academic Publishers, ISBN 978-1-61209-523-3,
2011.


\bibitem{Levenshtein}
V. I. Levenshtein, ``Bounds for packings of metric spaces and some of
their applications,'' (in Russian) Probl. Cybern., vol. 40, pp. 43-110,
1983.

\bibitem{LGGZ14}
S. Li, F. Gao, G. Ge, and S. Zhang, ``Deterministic sensing matrices arising from
near orthogonal systems,'' IEEE Trans. Inf. Theory, vol. 60, no. 4,
pp. 2291-2302, Apr. 2014.


  \bibitem{Lidl} R. Lidl and  H. Niederreiter,
        Finite Fields,
        Encyclopedia of Mathematics and its Applications 20, 1997.

\bibitem{Menichetti}
G. Menichetti, ``On a Kaplansky conjecture concerning three-dimensional division algebras over a finite
field,'' J. Algebra, vol. 47, no. 2, pp. 400¨C410, 1977.


\bibitem{NybergKnudsen}
K. Nyberg and  L.R. Knudsen, ``Provable security against differential cryptanalysis,'' In: Advances in
Cryptology¡ªCRYPTO'92, Santa Barbara, CA, 1992. Lecture Notes in Comput. Sci., vol. 740, pp.
566-574. Springer, Berlin (1993).

\bibitem{PottSchmidtZhou}
A. Pott, K. Schmidt and Y. Zhou, ``Semifields, Relative Difference Sets,
and Bent Functions,'' In H. Niederreiter, A. Ostafe, D. Panario, and A. Winterhof, editors, Algebraic
Curves and Finite Fields, Cryptography and Other Applications. De Gruyter,
2014.

\bibitem{ScherrZieve}
 Z. Scherr and M.E. Zieve, ``Some Planar monomials in characteristic 2,'' Ann. Comb., vol. 18,  pp. 723-729, 2014.

\bibitem{SchmidtZhou}
K.-U. Schmidt and Y. Zhou, ``Planar functions over fields of characteristic two,''
J. Algebraic Comb., vol. 40, pp. 503-526, 2014.

\bibitem{Wan2003}
Z.X. Wan,  Lectures on finite fields and Galois rings. World Scientific Publishing Co., Inc., River Edge, 2003.

\bibitem{Welch1974}
L. Welch, ``Lower bounds on the maximum cross correlation of signals,''
IEEE Trans. Inf. Theory, vol. 20, no. 3, pp. 397-399, 1974.


\bibitem{WoottersFields}
W.K. Wootters and B.D. Fields, ``Optimal state-determination by mutually unbiased measurements,'' Ann.
Phys. vol. 191, no. 2, pp. 363-381, 1989.

\bibitem{XDM}
C. Xiang, C. Ding, and S. Mesnager, ``Optimal Codebooks from Binary Codes Meeting the
Levenshtein Bound,'' IEEE Trans. Inf. Theory, vol. 61, no. 12, pp. 6526-6535, 2015.

\bibitem{Zhou2013}
 Y. Zhou, ``$(2^n, 2^n, 2^n, 1)$-relative difference sets and their representations,'' J. Comb. Des., vol. 21, no. 12, pp. 563-584,
2013.

\bibitem{ZT2011}
Z. Zhou and X. Tang, ``New Nearly Optimal Codebooks from Relative Difference Sets,'' Adv. in Math. Communications,
vol. 5, no. 3, pp. 521-527, 2011.

\bibitem{ZhangFeng}
A.X. Zhang and K. Feng, ``Two classes of codebooks nearly meeting the
Welch bound,'' IEEE Trans. Inf. Theory, vol. 58, no. 4, pp. 2507-2511,
2012.

\bibitem{ZDL14}
Z. Zhou, C. Ding, and N. Li, ``New families of codebooks achieving
the Levenshtein bounds,'' IEEE Trans. Inf. Theory, vol. 60, no. 11, pp.
7382-7387, 2014.

\end{thebibliography}
\end{document}